\documentclass[lettersize,journal]{IEEEtran}
\IEEEoverridecommandlockouts
\usepackage{cite}
\usepackage[T1]{fontenc}
\usepackage{aecompl}
\usepackage[numbers,sort&compress]{natbib}
\usepackage{amsmath,amssymb,amsfonts}
\usepackage{float}
\usepackage{booktabs}
\makeatletter
\renewcommand{\maketag@@@}[1]{\hbox{\m@th\normalsize\normalfont#1}}%
\makeatother
\usepackage{algpseudocode}
\usepackage{graphicx}
\usepackage{caption}
\allowdisplaybreaks[4]
\usepackage{subfigure}
\usepackage{textcomp}
\usepackage{xcolor}
\usepackage[ruled,linesnumbered]{algorithm2e}
\makeatletter
\newcommand{\removelatexerror}{\let\@latex@error\@gobble}
\makeatother
\usepackage{
    caption}
\usepackage{float}
\usepackage{caption}
\usepackage{graphicx}
\usepackage{textcomp}
\usepackage{xcolor}
\usepackage{bm}
\usepackage{geometry}
\usepackage{setspace}
\usepackage{amsmath}
\usepackage{amsthm}
\usepackage{setspace}
\def\d{\,\mathrm{d}}

\allowdisplaybreaks[4]
\newtheorem{theorem}{Theorem}
\newtheorem{corollary}{Corollary}
\newtheorem{lemma}{Lemma}
\newtheorem{myDef}{Definition}
\def\BibTeX{{\rm B\kern-.05em{\sc i\kern-.025em b}\kern-.08em
    T\kern-.1667em\lower.7ex\hbox{E}\kern-.125emX}}
\begin{document}

\title{\!\!Performance Optimization in RSMA-assisted Uplink xURLLC IIoT Networks with Statistical QoS Provisioning}
\author{Yuang Chen, \IEEEmembership{Graduate Student Member, IEEE}, Hancheng Lu, \IEEEmembership{Senior Member, IEEE}, \\ Chang Wu, \IEEEmembership{Graduate Student Member, IEEE}, Langtian Qin, and Xiaobo Guo.
\thanks{\setlength{\baselineskip}{1.5\baselineskip} Yuang Chen, Hancheng Lu, Chang Wu, and Langtian Qin are with the University of Science and Technology of China, Hefei 230027, China (email: yuangchen21@mail.ustc.edu.cn; hclu@ustc.edu.cn; \{changwu, qlt315\}@mail.ustc.edu.cn). Hancheng Lu is also with the Institute of Artificial Intelligence, Hefei Comprehensive National Science Center, Hefei 230088. Xiaobo Guo is with the National Key Laboratory of Advanced Communication Networks, Academy for Network $\&$ Communications of CETC, Shijiazhuang, China (email: umbrac177@163.com).}}
\maketitle{}

\begin{abstract}
Industry 5.0 and beyond networks have driven the emergence of numerous mission-critical applications, prompting contemplation of the neXt-generation ultra-reliable low-latency communication (xURLLC). To guarantee low-latency requirements, xURLLC heavily relies on short-blocklength packets with sporadic arrival traffic. As a disruptive multi-access technique, rate-splitting multiple access (RSMA) has emerged as a promising avenue to enhance quality of service (QoS) and flexibly manage interference for next-generation communication networks. In this paper, we investigate an innovative RSMA-assisted uplink xURLLC industrial internet-of-things (IIoT) (RSMA-xURLLC-IIoT) network. To unveil reliable insights into the statistical QoS provisioning (SQP) for our proposed network with sporadic arrival traffic, we leverage stochastic network calculus (SNC) to develop a dependable theoretical framework. Building upon this theoretical framework, we formulate the SQP-driven short-packet size maximization problem and the SQP-driven transmit power minimization problem, aiming to guarantee the SQP performance to latency, decoding, and reliability while maximizing the short-packet size and minimizing the transmit power, respectively. By exploiting Monte-Carlo methods, we have thoroughly validated the dependability of the developed theoretical framework. Moreover, through extensive comparison analysis with state-of-the-art multi-access techniques, including non-orthogonal multiple access (NOMA) and orthogonal multiple access (OMA), we have demonstrated the superior performance gains achieved by the proposed RSMA-xURLLC-IIoT networks.
\end{abstract}
\begin{IEEEkeywords}
next-generation ultra-reliable and low-latency communications (xURLLC), stochastic network calculus, rate-splitting multiple access (RSMA), industrial internet-of-things.
\end{IEEEkeywords}

\vspace{-0.95em}

\section{Introduction}
\par With the rapid evolution of mobile wireless network technologies, there has been a proliferation of numerous high-stack control and mission-critical emergent industrial applications, such as metaverse, robot control, and industrial automation \cite{park2022extreme,she2021tutorial}. These industrial applications have sparked unprecedented stringent quality of service (QoS) requirements, encompassing flexible interference management, robust multi-access, impeccable reliability, and resilient low latency, prompting people's envisioning for neXt-generation ultra-reliable and low-latency communications (xURLLC) \cite{park2022extreme,10382447,Yuang2023When,chen2024enhancing,10021621}. xURLLC stands as a pivotal cornerstone within the wireless IIoT ecosystems, the heightened anticipation regarding xURLLC has accelerated the development and customized standardization of Industry 5.0 and beyond networks \cite{xian2023advanced,she2021tutorial,park2022extreme}.

\vspace{-0.1em}

\par In IIoT scenarios requiring precision critical to mission success, base station (BS) and IIoT devices are typically deployed within confined environments. These IIoT devices transmit short-packet data, awakened periodically by internal events and sporadically triggered by external events, via the uplink to BS \cite{khoshnevisan20195g}. The BS seamlessly forwards these data to the data fusion center (DFC), where actionable insights are promptly derived and executed. Effectively meeting xURLLC's low-latency requirements imperatives necessitates the organization of extensive short-packet data transmissions within dynamic wireless networks \cite{chen2024enhancing,Yuang2023When,10382447}. However, under Shannon regimes, the decoding error probability (DEP) associated with long blocklength channel codes virtually approaches zero \cite{polyanskiy2010channel,yang2014quasi,10382447}, rendering the maximum achievable rate (MAR) inadequate for characterizing xURLLC's short-packet data transmissions. To address this theoretical limitation, finite blocklength coding (FBC) theory has been advanced, revealing MARs for short-packet data transmissions and elucidating the intricate tradeoff between DEP and target latency \cite{polyanskiy2010channel,yang2014quasi,10382447,Yuang2023When}. This tradeoff is influenced not only by blocklength but also by factors such as transmit power, inter-user interference, and transmission rate. Therefore, to guarantee xURLLC's QoS requirements, there arises an urgent need for robust multi-access techniques that provide more flexible interference management, higher spectrum efficiency (SE), and broader coverage.

\vspace{-0.1em}

\par Rate-splitting multiple access (RSMA), positioned as a candidate technique for next-generation wireless networks, has garnered widespread attention owing to its flexible interference management, effective resource sharing, and superior SE \cite{mao2022rate, clerckx2023primer, mishra2022rate}. In uplink RSMA scenarios, users split their transmitted messages into multiple streams, allocating suitable transmit power to each stream based on specific power allocation schemes \cite{rimoldi1996rate}. At the BS receiver, successive interference cancellation (SIC) is performed to decode each stream in designated orders, subsequently combining and reconstructing these streams. Through message splitting, RSMA facilitates flexible management of inter-user interference in the uplink without resorting to time sharing among users to achieve capacity \cite{mao2022rate, clerckx2023primer, mishra2022rate}. Remarkably, non-orthogonal multiple access (NOMA) falls within the subset of RSMA, with the distinction that NOMA does not involve the splitting of users' messages. RSMA's message splitting mechanism naturally aligns with core services (e.g., xURLLC) with sporadic access request behavior.

\vspace{-1.0em}

\subsection{Motivations and Challenges}

\vspace{-0.3em}

\par Although RSMA has demonstrated promise in enhancing the QoS in next-generation wireless networks, recent research has predominantly concentrated on downlink multi-antenna and multi-access scenarios under Shannon regimes with long blocklengths \cite{li2024synergizing,dizdar2024rate,lei2024secure}. Studies reveal that uplink RSMA can realize the optimal rate region of a $2M$-user Gaussian Multiple Access Channel (MAC) by exploiting up to $4M-1$ virtual point-to-point Gaussian channels generated through message splitting, thereby providing superior reliability, sum-throughput, and lower latency compared to non-orthogonal multiple access (NOMA) \cite{10430407} and orthogonal multiple access (OMA) \cite{rimoldi1996rate}. Furthermore, most studies on uplink RSMA are conducted under Shannon regimes with long blocklengths \cite{khisa2024power, sarker2023uplink,9852986}, while the exploration of RSMA-assisted xURLLC under FBL regimes is still in its infancy, accompanied by various challenges.

\vspace{-0.1em}

\par On the one hand, the stringent low-latency demands of IIoT devices necessitate xURLLC to operate under FBL regimes, underscoring the need to investigate the DEP and total throughput of uplink RSMA-assisted xURLLC networks. However, RSMA's flexible interference management, involving message splitting and combining, potentially introduces complexity to networks and poses challenges for seamless xURLLC integration \cite{chen2024enhancing, clerckx2023primer, mao2022rate}. On the other hand, previous research endeavors on xURLLC have predominantly focused on deterministic QoS provisioning (DQP) mechanisms, which are insufficient for analyzing the QoS requirements of extreme and rare events \cite{Yuang2023When,park2022extreme,10382447,chen2023streaming,she2021tutorial,bennis2018ultrareliable}. Specifically, xURLLC demands $99.9999\%$ reliability, sub-1 \emph{ms} latency, statistical QoS provisioning (SQP), and relies on short-blocklength packets with sporadic arrival traffic, presenting unparalleled challenges to the existing Industrial Internet \cite{Yuang2023When,park2022extreme}. Moreover, given the highly time-varying nature of wireless fading channels, DQP fails to accurately capture the queueing behavior of xURLLC traffic \cite{Yuang2023When,10382447,chen2023streaming}. This deviates from the fundamental design principles mandated by xURLLC, which involve analyzing the tail distribution of stringent QoS, underscores the need for a paradigm shift. Consequently, a thorough understanding of statistical metrics associated with rare and extreme events, known as the tail distribution, is essential to accommodate xURLLC's stringent QoS requirements \cite{Yuang2023When,10382447,chen2023streaming}.

\vspace{-0.2em}

\par Although embracing the SQP mechanism for xURLLC appears more feasible, there is a dearth of relevant studies in this area, leaving the underlying principles enigmatic \cite{she2021tutorial,park2022extreme, 10382447,Yuang2023When,bennis2018ultrareliable,chen2023streaming}. Currently, SQP theory championed by stochastic network calculus (SNC), has gained extensive traction \cite{al2014network,fidler2010survey,fidler2014guide}. SNC has evolved into a prominent methodology providing dependable theoretical insights into the SQP performance of latency-sensitive services \cite{Yuang2023When,10382447,chen2023streaming}. Nevertheless, current endeavors fall short of providing sufficient guidance for xURLLC's SQP, and existing state-of-the-art solutions for xURLLC operate independently, lacking direct compatibility and seamless integration \cite{singh2023rsma,ou2022resource,wang2023flexible,kurma2022urllc,muhammad2021mission,lien2022intelligent}. Therefore, tailoring an innovative RSMA-assisted transmission framework for xURLLC looms large.

\vspace{-1em}

\subsection{Related Work}
\par State-of-the-art research endeavors have extensively focused on uplink RSMA under Shannon regimes with long blocklengths \cite{khisa2024power, sarker2023uplink,9852986,10330667}, and relatively little research has focused on uplink RSMA under FBL regimes \cite{9970313}, let alone considering the SQP-based tail analysis for uplink RSMA-assisted xURLLC networks. The authors in \cite{khisa2024power} have investigated a two-user uplink cooperative RSMA (C-RSMA) scheme, which jointly optimizes beamforming and device transmit power. In \cite{sarker2023uplink}, authors have tackled the optimization problem of maximizing the minimum SE of uplink RSMA-assisted user-centric massive MIMO networks. To enhance SE and user fairness, the authors in \cite{9852986} have proposed a novel C-RSMA scheme for uplink user cooperation. Additionally, the authors in \cite{10330667} have studied the optimization problems of power allocation and decoding order, and have proposed a low-complexity user-pair-based resource allocation algorithm to minimize the maximum latency. To study the impact of target rate and blocklength on the DEP and throughput performance, the authors in \cite{9970313} have investigated the performance of uplink RSMA under FBL regimes.

\par xURLLC underscores extreme and rare events such as 99.9999\% reliability and sub-1 \emph{ms} latency \cite{park2022extreme, 10323296}. Currently, some relevant studies have explored the statistical QoS provisioning for xURLLC \cite{Yuang2023When,10382447,chen2023streaming,9887634,10220199,10445487}. In \cite{Yuang2023When}, we have proposed a NOMA-assisted uplink xURLLC network and analyzed the tail distributions of latency, age-of-information (AoI), and reliability of xURLLC leveraging SNC theory. To explore xURLLC's fundamentals and performance tradeoffs, we have developed an xURLLC-enabled massive MIMO network and proposed a theoretical framework for statistical QoS provisioning analysis for xURLLC by leveraging and promoting SNC theory \cite{10382447}. In \cite{chen2023streaming}, the authors have proposed a 360$^{\circ}$ virtual reality (VR) streaming architecture with statistical QoS provisioning by developing an SNC-based SQP theoretical framework from delay and rate perspectives. To fulfill QoS guarantees for URLLC slicing, the authors in \cite{9887634} have proposed an SNC-based URLLC slicing allocation scheme, which effectively ensures the QoS violation probability remains consistently lower than the target value over the long term with a certain probability. To investigate the violation probability of end-to-end (E2E) latency for target services in IIoT, the authors in \cite{10220199} have introduced an SNC-based cascaded theorem to evaluate the service capability of xURLLC networks and studied the service processes across various wireless fading channels.

\vspace{-1.2em}

\subsection{Main Contributions}

\vspace{-0.1em}

\par To effectively overcome the aforementioned challenges, we have developed an innovative RSMA-assisted uplink xURLLC IIoT (RSMA-xURLLC-IIoT) network architecture under FBL regimes and imperfect channel-state-information (CSI) scenarios. To accurately capture the queueing behavior of sporadic xURLLC traffic, we have proposed an SNC-based SQP (SNC-SQP) theoretical framework to analyze the tail distribution of xURLLC's QoS requirements. Building upon this theoretical framework, we have formulated two SQP-driven optimization problems tailored for the developed RSMA-xURLLC-IIoT network architecture and proposed a low-complexity algorithm to tackle them. Extensive simulations demonstrate the dependability of the proposed theoretical framework, and comprehensive comparisons validate the superior performance of our developed RSMA-xURLLC-IIoT network architecture. The primary contributions of this paper are summarized as follows:

\vspace{-0.2em}

\begin{itemize}
  \item We have developed an innovative RSMA-xURLLC-IIoT network architecture, which seamlessly bridges the integration gaps between RSMA and xURLLC under FBL regimes and imperfect CSI scenarios, providing flexible interference management, robust multi-access, impeccable reliability, and resilient low latency.

  \item Leveraging SNC theory, we have proposed a dependable SNC-SQP theoretical framework. In particular, a novel terminology named statistical delay violation probability (SDVP) is introduced to access the SQP performance of xURLLC's delay. Subsequently, we derive a closed-form upper-bounded expression for SDVP termed UB-SDVP to quantify SDVP accurately. Additionally, we have derived various closed-form expressions for the probability density functions (PDFs) of wireless channels involved in our developed RSMA-xURLLC-IIoT network architecture.

  \item Building upon the proposed theoretical framework, we have delved into two critical optimization problems. The first one addresses the SQP-driven short-packet size maximization problem, which holds extraordinary significance for the IIoT since it guarantees not only the execution of complex machine instructions but also the accomplishment of coherent mission-critical tasks. The second problem deals with the SQP-driven transmit power minimization, which is also of great importance, particularly with the application prospect of massive xURLLC networks and battery-constrained IIoT devices. Then, we have proposed a low-complexity three-step sequential optimization algorithm (TSOO) for their efficient resolutions.

  \item Extensive simulations have reliably demonstrated the effectiveness of our developed RSMA-xURLLC-IIoT network architecture. Comprehensive comparisons with results from Monte-Carlo methods rigorously validate the dependability of the proposed SNC-SQP theoretical framework. Furthermore, through comparisons with prevalent NOMA and OMA, we have further substantiated the remarkable performance gains achieved by our developed network architecture.
\end{itemize}

\vspace{-0.6em}

\par The remainder of this paper is organized as follows. In Sec. II, the RSMA-xURLLC IIoT network architecture is investigated. In Sec. III, the SNC-SQP theoretical framework is developed. In Sec. IV, two optimal SQP-driven optimization problems are formulated and addressed, followed by the performance evaluation and comparisons. Finally, a conclusion of this paper is given.

\vspace{-2em}

\section{RSMA-Assisted Uplink xURLLC IIoT Networks Architecture}

\vspace{-0.4em}

\par As illustrated in Fig. \ref{fig1}, we consider an RSMA-assisted uplink IIoT network architecture accessed with multiple IIoT devices, where $2M$ single-antenna IIoT devices simultaneously transmit short-packet data to a base station (BS) equipped with a single antenna. These $2M$ IIoT devices are randomly grouped into $M$ pairs, each pair comprising two IIoT devices. The available spectrum is evenly divided into $M$ orthogonal subchannels, with each assigned to an IIoT device pair. Without loss of generality, we focus on a specific IIoT device pair designated as $\mathcal{U} \!=\! \left\{1,2\right\}$ \footnote{To facilitate the analysis of the SQP performance for the developed RSMA-xURLLC-IIoT network architecture, without loss of generality, this paper considers the single-antenna setup, which is consistent with the most prevailing uplink RSMA-assisted IIoT systems \cite{9970313,10330667, 9852986, singh2023rsma}. While the multi-antenna setup enhances the performance of our developed RSMA-xURLLC-IIoT network architecture, it leads to the increased complexity of IIoT system design, which is also beyond the scope of this article. However, it is worth noting that the network architecture considered in this paper can be reasonably extended to multi-antenna scenarios, with recent references expected to provide valuable insights into this aspect\cite{khisa2024power,sarker2023uplink,ou2022resource,10382447}.}.

\vspace{-1em}

\begin{figure*}[t]
\vspace{-0.5em}
\centering
\includegraphics[scale=0.7]{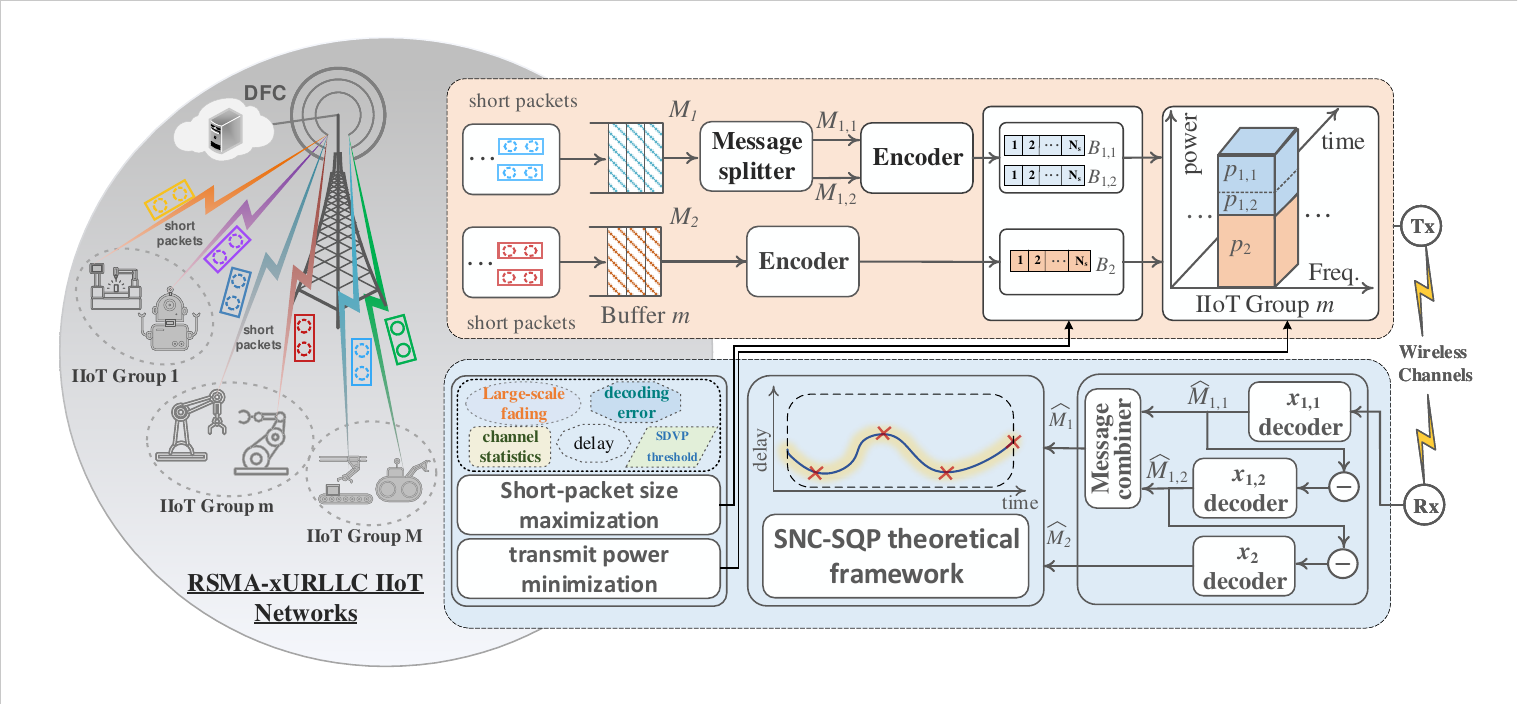}
\caption{The RSMA-assisted uplink xURLLC IIoT network architecture.}
\label{fig1}
\vspace{-0.5em}
\end{figure*}

\subsection{RSMA for Uplink xURLLC with Imperfect CSI}

\vspace{-0.2em}

\par At the transmitter side, the message from IIoT device-$u$ at time slot $t$ is represented as $M_{u}(t), u \!\in\! \mathcal{U}$. Following RSMA principles \cite{clerckx2023primer,mao2022rate}, the message $M_{1}(t)$ from IIoT device-$1$ is split into two components, denoted as $M_{1,k}(t), k \!\in\! \left\{1,2\right\}$. The message $M_{q}(t)$ can be first encoded into the stream $x_{q}(t)$, where $q \in \mathcal{Q} \triangleq \{(1,1),(1,2),(2)\}$. These streams are mapped into finite-length codewords and transmitted over wireless fading channels. The received signal at the BS can be given as follows:

\begin{equation}\label{e1}
\begin{aligned}
  y(t) & =\left(\sqrt{p_{1,1}}x_{1,1}(t) + \sqrt{p_{1,2}}x_{1,2}(t)\right)g_{1}(t),\\
       &\quad + \sqrt{p_{2}}x_{2}(t)g_{2}(t) + n(t),
\end{aligned}
\end{equation}
where $p_{q}$ is the uplink transmit power of $x_{q}$, $g_{u}(t) \!\!=\!\! \sqrt{\zeta_{u}(t)}h_{u}(t)$ represents the channel coefficient of IIoT device-$\!u \!\in\! \mathcal{U}$, and $n(t) \! \sim \! \mathcal{CN}\left(0,\sigma_{n}^{2}\right)$ is the Additive White Gaussian Noise. $\zeta_{u}(t)$ and $h_{u}(t)$ denote the large-scale fading coefficient and small-scale fading coefficient, respectively. 

\par In realistic industrial scenarios, channel estimation is essential since IIoT devices do not have perfect prior knowledge of channel states, and the actual channel coefficient $h_{u}(t)$ can not be perfectly known. We consider each IIoT device-$u$ transmits a known orthogonal pilot sequence of length $N_{p,u}$ to the BS at the beginning of each time slot. As a result, the BS can obtain the minimum mean squared error (MMSE) channel estimation $\hat{h}_{u}(t)$ for IIoT device-$u$. Given the obtained $\hat{h}_{u}(t)$, the $h_{u}(t)$ can be given as

\begin{equation}\label{e2}
   h_{u}(t) = \hat{h}_{u}(t) + \tilde{e}_{u}(t), u \in \mathcal{U},
\end{equation}
where $\tilde{e}_{u}(t) \! \sim \! \mathcal{CN}\left(0,\sigma_{e,u}^{2}\right)$, with $\sigma_{e,u}^{2} \!=\! \frac{1}{1+\bar{\gamma}_{p,u}(t)N_{p,u}}$, and $\bar{\gamma}_{p,u}(t)$ denotes the actual signal-noise-ratio (SNR) of IIoT device-$u$ during the training phase, which is constant and known at both transmitter and receiver. The channel estimation $\hat{h}_{u}$ is Gaussian distributed as $\hat{h}_{u} \!\sim\! \mathcal{CN}\left(0,\rho_{u}^{2}\right)$ with $\rho_{u}^{2} \!=\! \bar{\gamma}_{p,u}N_{p,u}/\left(1 + \bar{\gamma}_{p,u}N_{p,u}\right)$. Then, the actual SNR $\Gamma_{q} = \bar{\gamma}_{q}|h_{u}|^{2}$ can be given as follows:

\begin{equation}\label{e3}
   \Gamma_{q} = \bar{\gamma}_{q}|\hat{h}_{u}|^{2} + 2\bar{\gamma}_{q}|\hat{h}_{u}|\Re\left\{e^{-j\angle\left(\hat{h}_{u}\right)}\tilde{e}_{u}\right\} + \bar{\gamma}_{q}|\tilde{e}_{u}|^{2},
\end{equation}
where operators $\Re\{\cdot\}$ and $\angle\{\cdot\}$ represent the real part and phase of the complex number, respectively. The term $\Re\left\{e^{-j\angle\left(\hat{h}_{u}\right)}\tilde{e}_{u}\right\}$ follows a Gaussian distribution with variance $\sigma_{e,u}^{2}/2$, and the term of estimation error $\bar{\gamma}_{q}|\tilde{e}_{u}|^{2}$ becomes quite small after sufficient pilot training, i.e., $|\tilde{e}_{u}|\ll|\hat{h}_{u}|$, which can be reasonably neglected in (\ref{e3}). In this case, the distribution of $\Gamma_{q}$ can be approximated as $\Gamma_{q} \sim \mathcal{N}\left(\widehat{\gamma}_{q},\widehat{\sigma}_{q}\right)$, where $\widehat{\gamma}_{q} = \bar{\gamma}_{q}|\hat{h}_{u}|^{2}$, $\widehat{\sigma}_{q} = 2\bar{\gamma}_{u}^{2}|\hat{h}_{u}|^{2}\sigma_{e,u}^{2}$, and $\bar{\gamma}_{q} = p_{q}\zeta_{u}/\sigma_{n}^{2}$ denotes the average SNR of stream $x_{q}$. If $q \in \big\{(1,1),(1,2)\big\}$, then $u = 1$; otherwise, $u = 2$. The blocklength used for short-packet communications and channel estimation is typically extremely limited.
We consider that the blocklength of each time slot is $N_{0}$, then only $N_{d} \!=\! N_{0} - N_{p,1} - N_{p,2}$ channel uses (CUs) available for short-packet data transmission.

\par At the receiver side, we consider that the streams are decoded in the order of $x_{1,1} \!\rightarrow \!x_{1,2} \!\rightarrow\! x_{2}$\footnote{The decoding order can be  categorized into six cases, including (1) $x_{1,1} \!\rightarrow \! x_{2} \! \rightarrow \! x_{1,2}$, (2) $x_{1,2} \!\rightarrow \!x_{2} \!\rightarrow\! x_{1,1}$, (3) $x_{1,1} \!\rightarrow \!x_{1,2} \!\rightarrow\! x_{2}$, (4) $x_{1,2} \!\rightarrow\! x_{1,1} \!\rightarrow\! x_{2}$, (5) $x_{2} \!\rightarrow\! x_{1,1} \!\rightarrow\! x_{1,2}$, (6) $x_{2} \!\rightarrow\! x_{1,2} \!\rightarrow\! x_{1,1}$. Without loss of generality, we select the decoding order (3) in this paper.}. Accordingly, $x_{1,1}(t)$ is first decoded by treating $x_{1,2}(t)$ and $x_{2}(t)$ as noise. Then, the SINR of $x_{1,1}(t)$ is represented as

\begin{equation}\label{e4}
  \widehat{\Gamma}_{1,1}(t) = \frac{\Gamma_{1,1}(t)}{\Gamma_{1,2}(t) + \Gamma_{2}(t) + 1}.
\end{equation}

\par If $\!x_{1,1}(t)\!$ is decoded successfully, it can be reconstructed into $\!\widehat{M}_{1,1}(t)$, and then removed from the received signal $y(t)$. Secondly, $x_{1,2}(t)$ is decoded by treating $\!x_{2}(t)\!$ as noise. Thus, the SINR of $x_{1,2}(t)$ can be expressed as

\begin{equation}\label{e5}
   \widehat{\Gamma}_{1,2}(t) = \frac{\Gamma_{1,2}(t)}{\Gamma_{2}(t) + 1}.
\end{equation}

\par If $x_{1,2}(t)\!$ is decoded successfully, it can also be reconstructed into $\!\!\widehat{M}_{1,2}(t)$, and removed from $y(t)$. Using a message combiner, the estimated message of IIoT device-$1$ is $\!\widehat{M}_{1}(t) \!=\! \widehat{M}_{1,1}(t) \!+\! \widehat{M}_{1,2}(t)$. Finally, the SINR of $x_{2}(t)$ can be given as follows:
\begin{equation}\label{e6}
  \widehat{\Gamma}_{2}(t) = \Gamma_{2}(t).
\end{equation}

\par If $x_{2}(t)$ is decoded successfully, it can be reconstructed into $\widehat{M}_{2}(t)$.

\vspace{-1.0em}

\subsection{Short-packet Data Communications Model}
\vspace{-0.3em}
\par Considering the short-packet size of $x_{q}$ to be $B_{q}$ bits, typically $B_{q} \!\leq\! 500$ bits, $q \! \in \!\mathcal{Q}$. Then, the achievable coding rate can be expressed as $R_{q} = \frac{B_{q}}{N_{d}}$ (bpcu). Following the finite blocklength coding (FBC) theory \cite{10382447,Yuang2023When,polyanskiy2010channel}, the decoding error probability can be given by

\vspace{-1.3em}
\begin{equation}\label{e7}
   \begin{aligned}
    \bar{\boldsymbol{\epsilon}}_{q}\left(\boldsymbol{\mathrm{B}},\boldsymbol{\mathrm{p}}\right) & = \mathbb{E}_{\widehat{\Gamma}_{q}}\left[Q\left(\frac{\ln\left(1 + \widehat{\Gamma}_{q}(t)\right)-\frac{B_{q}}{N_{d}}\ln2}{\sqrt{\frac{1}{N_{d}}\left(1 - \frac{1}{(1 + \widehat{\Gamma}_{q})^{2}}\right)}}\right)\right]\\
    & \triangleq \mathbb{E}_{\widehat{\Gamma}_{q}}\bigg[Q\big(f(\boldsymbol{\mathrm{B}},\boldsymbol{\mathrm{p}})\big)\bigg].
   \end{aligned}
\end{equation}
where $Q(\cdot)$ denotes the Q-function, which can be expressed by $Q(x) = \int_{x}^{\infty} \frac{1}{\sqrt{2\pi}}e^{-x^{2}/2}\mathrm{d}x$. According to RSMA principles \cite{chen2024enhancing,clerckx2023primer,mao2022rate}, since the rate of the IIoT device-$1$ $R_{1}$ is split, the achievable data rate of $x_{1,1}(t)$ and $x_{1,2}(t)$ can be represented as $R_{1,1}(t) \!=\! \alpha R_{1}(t)$ and $R_{1,2}(t) \!=\! (1-\alpha)R_{1}(t)$, respectively, where $0 \leq \alpha \leq 1$ denotes the rate allocation ratio \cite{clerckx2023primer,mao2022rate}.


\vspace{-0.5em}

\section{SNC-Based SQP Theoretical Framework}

\vspace{-0.2em}

\par In this section, we develop a dependable SNC-SQP theoretical framework. In particular, we introduce a novel terminology termed SDVP to characterize the tail distribution of delays. We derive the closed-form expressions for the PDFs of $\widehat{\Gamma}_{1,1}$, $\widehat{\Gamma}_{1,2}$, and $\widehat{\Gamma}_{2}$, and subsequently induced the closed-form expression of UB-SDVP.

\vspace{-0.5em}

\subsection{Statistical Delay Violation Probability}

\par To facilitate system-level analysis, we consider a statistical QoS-driven queueing system that employs the first-come-first-serve (FCFS) policy \cite{10382447,chen2023streaming,fidler2014guide}. The cumulative arrival, service, and departure processes of $x_{q}$ within the time interval $[s,t)$ are given as $A_{q}(s,t) \!\triangleq\! \sum_{i = s}^{t-1}a_{q}(i)$, $S_{q}(s,t) \!\triangleq\! \sum_{i = s}^{t-1}s_{q}(i)$, and $D_{q}(s,t) \!\triangleq\! \sum_{i = s}^{t-1}d_{q}(i)$, respectively. Here, $a_{q}(i)$ and $s_{q}(i)$ represent the amount of short-packet data generated by $x_{q}(i)$ and transmitted over wireless channels at time slot $i$, respectively, while $d_{q}(i)$ indicates the amount of short-packet data of $x_{q}(i)$ that is successfully received at time slot $i$. Then, to facilitate the analysis of SQP performance, we introduce a succinct and novel operator by leveraging moment generating function (MGF)-based SNC theory, which is termed \emph{min-deconvolution} $\widehat{\oslash}$ as follows:

\begin{myDef}
  (min-deconvolution $\widehat{\oslash}$): The min-deconvolution between $A_{q}(s,t)$ and $S_{q}(s,t)$ for stream $x_{q}, q\!\in\!\mathcal{Q}$ can be given as follows:
  \begin{equation}\label{e8}
     \mathrm{\mathbf{M}}_{A_{q}\widehat{\oslash}S_{q}}\!(\theta_{q},s,t) = \!\! \sum\limits_{v=1}^{\min\{s,t\}}\!\!\!\!\mathbb{M}_{A_{q}}(\theta_{q},v,t)\cdot \overline{\mathbb{M}}_{S_{q}}(\theta_{q},v,s),
  \end{equation}
  where the parameter $\theta_{q} \!\! >\!\! 0$ denotes the statistical QoS exponent of $x_{q}$. $\mathbb{M}_{A_{q}}(\theta_{q},v,t)\!$ and $\overline{\mathbb{M}}_{S_{q}}(\theta_{q},v,t)\!$ refer to the MGF of $A_{q}$ and inverse-MGF of $S_{q}$, respectively \footnote{ In SNC theory \cite{10382447,fidler2014guide}, $\theta_{q}$ is exploited to characterize the decay rate of the queue length for statistical QoS-driven queueing systems. A larger $\theta_{q}$ corresponds to $x_{q}$ having stricter statistical QoS requirements. Conversely, a smaller $\theta_{q}$ implies looser statistical QoS requirements \cite{10382447,fidler2014guide}. Given a random process $U(s,t), 0\leq s\leq t$, the MGF of $U(s,t)$ is denoted as $\mathbb{M}_{U}(\theta,s,t) = \mathbb{E}\left[e^{\theta U(s,t)}\right]$, and the inverse-MGF of $U(s,t)$ is $\overline{\mathbb{M}}_{U}(\theta,s,t) = \mathbb{E}\left[e^{-\theta U(s,t)}\right]$ \cite{chen2023statistical,fidler2014guide}.}.
\end{myDef}

\par Based on the operator \emph{min-deconvolution}, we have the following \emph{Theorem 1}:

\begin{theorem}
  Given $\!A_{q}(s,t)$ and $S_{q}(s,t)$, the SDVP of $x_{q}, q\!\in\!\mathcal{Q}$ can be expressed as follows:
  \vspace{-0.4em}
  \begin{equation}\label{e9}
     \mathbb{P}\left(W_{q}(t) \!>\! w_{q}^{th}\right) \! \leq \! \inf_{\theta_{q} > 0} \mathrm{\mathbf{M}}_{A_{q}\widehat{\oslash}S_{q}}(\theta_{q},t+w_{q}^{th},t),
     \vspace{-0.8em}
  \end{equation}
  where $W_{q}(t)$ and $w_{q}^{th}$ denote the actual delay and target delay of $x_{q}$, respectively.
\end{theorem}

\begin{proof}
   The proof of \textbf{Theorem 1} is given in Appendix A.
\end{proof}

\vspace{-0.5em}

\subsection{The Closed-Form Expression of UB-SDVP}

\par The promotion of MGF-SNC theory in \textbf{Theorem 1} provides a crucial theoretical foundation for further deriving the closed-form expression of UB-SDVP for the developed RSMA-xURLLC-IIoT network architecture. We consider $A_{q}(s,t)$, $S_{q}(s,t)$, and $D_{q}(s,t)$ have independent and identically distributed (i.i.d.) increments. Then, the MGF and the inverse-MGF of $A_{q}(s,t)$ and $S_{q}(s,t)$ can be respectively expressed as follows:

\vspace{-1em}

\begin{subequations}\label{t1}
   \begin{align}
      & \mathbb{M}_{A_{q}}\left(\theta_{q},s,t\right) = \left(\mathbb{E}\left[e^{\theta_{q} a_{q}}\right]\right)^{t-s} = \left(\mathbb{M}_{a_{q}}\left(\theta_{q}\right)\right)^{t-s},\\
      & \overline{\mathbb{M}}_{S_{q}}\left(\theta_{q},s,t\right) = \left(\mathbb{E}\left[e^{\theta_{q} s_{q}}\right]\right)^{t-s} = \left(\overline{\mathbb{M}}_{s_{q}}\left(\theta_{q}\right)\right)^{t-s}.
   \end{align}
\end{subequations}

\par Combining \textbf{Definition 1} with (\ref{t1}a) and (\ref{t1}b), we can derive the closed-form expression of UB-SDVP, as stated in Theorem 2.

\begin{theorem}
  Given the target delay $w_{q}^{th}$, the SDVP of $x_{q}(t), q \!\in\! \mathcal{Q}$ can be upper bounded by
  \begin{equation}\label{e10}
     \mathbb{P}\left(W_{q}\left(t\right) \! \geq \! w_{q}^{th} \right) \! \leq \! \inf_{0 < \theta_{q} < \theta_{0}} \!\! \left\{\!\! \frac{\big(\overline{\mathbb{M}}_{s_{q}}\left(\theta_{q}\right)\big)^{w_{q}^{th}}}{1 \!-\! \mathbb{M}_{a_{q}}\left(\theta_{q}\right)\cdot \overline{\mathbb{M}}_{s_{q}}\left(\theta_{q}\right)} \!\!\right\},
  \end{equation}
  where $\mathbb{M}_{a_{q}}\left(\theta_{q}\right)\cdot \overline{\mathbb{M}}_{s_{q}}\left(\theta_{q}\right) < 1$ denotes the stability condition, and $\theta_{0} \!=\! \sup\left\{\!\theta_{q}\!:\!\mathbb{M}_{a_{q}}\!\!\left(\theta_{q}\right)\!\cdot\! \overline{\mathbb{M}}_{s_{q}}\!\!\left(\theta_{q}\right) \!<\! 1\right\}$.
\end{theorem}

\vspace{-1.0em}

\begin{proof}
According to \textbf{Definition 1}, the \emph{min-deconvolution} between $A_{q}\left(s,t\right)$ and $S_{q}\left(s,t\right)$ can be expressed by (\ref{t2}), where $\tau = \max\left\{0,s-t\right\}$. Substituting (\ref{t1}a) and (\ref{t1}b) into (\ref{t2}), respectively, we can derive the inequality (a). Based on inequality (a), equality (b) can be further deduced by taking variable replacement, i.e., $v = s - u$. Inequality (c) can be derived by scaling the upper bound of the summation sign in (b) from $s$ to $+\infty$. Finally, according to the relevant properties of the geometric series, when the stability condition $\mathbb{M}_{a_{q}}\left(\theta_{q}\right)\cdot \overline{\mathbb{M}}_{s_{q}}\left(\theta_{q}\right) < 1$ holds, equality (d) can be easily derived.

\begin{figure*}[t]
  \setstretch{0.90}
\centering
\hrulefill
  \begin{equation}\label{t2}
     \begin{aligned}
      \boldsymbol{\mathrm{M}}_{A_{q}\widehat{\oslash} S_{q}} \left(\theta_{q},s,t\right) &\overset{(a)}{\leq} \sum_{u=0}^{\min\left(s,t\right)} \left(\mathbb{M}_{a_{q}}\left(\theta_{q}\right)\right)^{t-u} \cdot \left(\overline{\mathbb{M}}_{s_{q}}\left(\theta_{q}\right)\right)^{s-u}
      \overset{(b)}{=} \left(\mathbb{M}_{a_{q}}\left(\theta_{q}\right)\right)^{t-s}\cdot \sum\limits_{v = \tau}^{s}\left(\mathbb{M}_{a_{q}}\left(\theta_{q}\right) \overline{\mathbb{M}}_{s_{q}}\left(\theta_{q}\right)\right)^{v}\\
      & \overset{(c)}{\leq} \left(\mathbb{M}_{a_{q}}\left(\theta_{q}\right)\right)^{t-s}\cdot \sum\limits_{v = \tau}^{\infty}\left(\mathbb{M}_{a_{q}}\left(\theta_{q}\right) \overline{\mathbb{M}}_{s_{q}}\left(\theta_{q}\right)\right)^{v} \overset{(d)}{=} \frac{\left(\mathbb{M}_{a_{q}}\left(\theta_{q}\right)\right)^{t-s}\cdot \left(\mathbb{M}_{a_{q}}\left(\theta_{q}\right) \overline{\mathbb{M}}_{s_{q}}\left(\theta_{q}\right)\right)^{\tau} }{1 - \mathbb{M}_{a_{q}}\left(\theta_{q}\right) \overline{\mathbb{M}}_{s_{q}}\left(\theta_{q}\right)}
     \end{aligned}
  \end{equation}
\hrulefill
\end{figure*}
So the proof of \textbf{Theorem 2} can be concluded.
\end{proof}

\par \textbf{Theorem 2} refers that further derivation of the inverse-MGF of the service process increment $s_{q}$ (i.e., $\overline{\mathbb{M}}_{s_{q}}\left(\theta_{q}\right)$) and the MGF of the arrival increment $a_{q}$ (i.e., $\mathbb{M}_{a_{q}}\left(\theta_{q}\right)$) is essential to determine the closed-form expression for UB-SDVP. In this case, we consider the arrival process follows a Poisson distribution, as follows:

\begin{equation}\label{e11}
   \mathbb{M}_{a_{q}}\!\left(\theta_{q}\right) \!=\! \sum\limits_{z=1}^{\infty}\!e^{z\theta_{q}}\frac{(\lambda_{q}^{\dag})^{z}e^{-\lambda_{q}^{\dag}}}{z!} \!=\! e^{\lambda_{q}^{\dag}\left(\!e^{\theta_{q}}-1\!\right)},
\end{equation}
where $\!\lambda_{q}^{\dag}\!$ denotes the average arrival rate of $\!A_{q}\!\left(s,t\right)$.

\par The inverse-MGF $\overline{\mathbb{M}}_{s_{q}}\!(\theta_{q})$ depends on $\widehat{\Gamma}_{q}(t)$, which varies from time slots, thus $\overline{\mathbb{M}}_{s_{q}}\!(\theta_{q})$ can be given by

\begin{equation}\label{e12}
       \overline{\mathbb{M}}_{s_{q}}\!\!\left(\theta_{q}\right)\!=\! \mathbb{E}_{\widehat{\Gamma}_{\!q}}\!\!\left[e^{-\theta_{q} N_{d} R_{q}}\right] \!+  \epsilon_{q}\left(B,N_{d},\boldsymbol{\mathrm{p}}\right)\big(1\!-\!\mathbb{E}_{\widehat{\Gamma}_{\!q}}\!\!\left[e^{-\theta_{q} N_{d} R_{q}}\right]\big),
\end{equation}
where $\mathbb{E}_{\widehat{\Gamma}_{q}}\left[\cdot\right]$ denotes the expectation operator of $\widehat{\Gamma}_{q}(t)$. From (\ref{e12}), the expression of p.d.f. for $\widehat{\Gamma}_{q}(t)$ is indispensable for determining $\overline{\mathbb{M}}_{s_{q}}\!\!\left(\theta_{q}\right)$, which motivates \textbf{Lemma 1} as follows:

\begin{lemma}
   Given the transmit power $\boldsymbol{\mathrm{p}} \!=\! \big[p_{1,1},p_{1,2},p_{2}\big]$, the PDF of $\widehat{\Gamma}_{q}$ can be given as follows:
   \begin{equation}\label{e13}
      \begin{aligned}
        & f_{\widehat{\Gamma}_{q}}\!(x)  = \\
        & \!\!\!\! \left\{
                            \begin{array}{ll}
                               \!\!\!\! \frac{1}{\widehat{\sigma}_{2}\sqrt{2\pi}}\exp\!\left\{\!-\frac{(x-\widehat{\gamma}_{2})^{2}}{2\widehat{\sigma}_{2}^{2}}\!\right\},\!\!\!&\hbox{if $\widehat{\Gamma}_{\!q} \! = \! \widehat{\Gamma}_{\!2}$;} \\
                               \!\!\!\! \int_{0}^{\infty}\!\!\!\frac{1 + y}{2\pi\sigma_{1,2}\sigma_{2}}\!\cdot\! \exp\!\left\{\!\!-\!\!\left(\!\frac{\left(x + xy - \widehat{\gamma}_{1,2}\right)^{2}}{2\widehat{\sigma}_{1,2}^{2}}\!+\!\frac{\left(y - \widehat{\gamma}_{2}\right)^{2}}{2\widehat{\sigma}_{2}^{2}}\!\right)\!\!\right\}\!\!\d y,\!\!\!&\hbox{if $\widehat{\Gamma}_{\!q} \! = \! \widehat{\Gamma}_{\!1,2}$;} \\
                               \!\!\!\! \int_{0}^{\infty}\!\!\!\frac{1 + y}{2\pi\widehat{\sigma}_{1,1}\widehat{\sigma}_{\sum}}\!\cdot\! \exp\!\left\{\!\!-\!\!\left(\!\!\frac{\left(x + xy - \widehat{\gamma}_{1,1}\right)^{2}}{2\widehat{\sigma}_{1,1}^{2}}\!+\!\frac{\left(y - \widehat{\gamma}_{\sum}\right)^{2}}{2\widehat{\sigma}_{\sum}^{2}}\!\right)\!\!\right\}\!\!\d y,\!\!\!&\hbox{if $\widehat{\Gamma}_{\!q} \! = \! \widehat{\Gamma}_{\!1,1}$,}
                            \end{array}
                          \right.
     \end{aligned}
   \end{equation}
   where $\widehat{\sigma}_{\sum}^{2} = \widehat{\sigma}_{1,2}^{2} + \widehat{\sigma}_{2}^{2}$ and $\widehat{\gamma}_{\sum} = \widehat{\gamma}_{1,2} + \widehat{\gamma}_{2}$.
\end{lemma}

\begin{proof}
 The proof of \textbf{Lemma 1} is given in Appendix B.
\end{proof}

\par Combining \textbf{Theorem 1}, \textbf{Theorem 2}, and \textbf{Lemma 1}, the closed-form expression of UB-SDVP can be finally derived.

\vspace{-0.5em}

\section{Problem Formulation And Solutions}
\par In this section, building upon the proposed SNC-SQP theoretical framework, we delve into two critical optimization problems within our developed RSMA-xURLLC-IIoT network architecture:
\begin{itemize}
  \item The SQP-driven short-packet size maximization problem, the study of which is of extraordinary significance for IIoT networks since it guarantees not only the execution of complex machine instructions but also the accomplishment of coherent mission-critical tasks.
  \item The SQP-driven transmit power minimization problem, which is also of great importance, particularly with the application prospect of massive xURLLC networks and battery-constrained IIoT devices.
\end{itemize}
\par Significantly, in both optimization problems, we jointly optimize the rate-allocation ratio, transmit power, and information bit number with the aim of maximizing short-packet size and minimizing transmit power for RSMA-xURLLC while guaranteeing the SQP performance, respectively.

\vspace{-1em}

\subsection{The SQP-driven Short-packet Size Maximization Problem}

\par Given the statistical QoS requirements $\left(\xi_{q}^{th},w_{q}^{th},\varepsilon_{q}^{th}\right)$, where $\xi_{q}^{th}$ denotes the threshold od SDVP, $w_{q}^{th}$ indicates the threshold of target delay, and $\varepsilon_{q}^{th}$ represents the threshold of decoding error probability. As a result, the SQP-driven short-packet size maximization problem can be formulated as

\vspace{-0.5em}

\begin{subequations}
   \begin{align}
      \mathcal{P}1:& \quad \max_{\{\boldsymbol{\mathrm{B}},\boldsymbol{\mathrm{p}},\alpha\}} \sum\limits_{q \in \mathcal{Q}} B_{q}, \label{e14a}\\
      \quad s.t. & \quad \mathbb{P}\!\left(W_{q}\!\left(t\right) \! \geq \! w_{q}^{th}\right) \! \leq \! \xi_{q}^{th},\ \forall q \in \mathcal{Q},\label{e14b}\\
                 & \quad \bar{\boldsymbol{\epsilon}}_{q}\left(\boldsymbol{\mathrm{B}},\boldsymbol{\mathrm{p}}\right) \! \leq \! \varepsilon_{q}^{th},\ \forall q \in \mathcal{Q},\label{e14c}\\
                 & \quad 0 \leq p_{q} \leq p_{max}, \ \forall q \in \mathcal{Q},\label{e14d}\\
                 & \quad B_{min} \leq B_{q} \leq B_{max}, \ \forall q \in \mathcal{Q},\label{e14e}\\
                 & \quad 0 \leq \alpha \leq 1\label{e14f},
   \end{align}
\end{subequations}
where $\boldsymbol{\mathrm{B}} \!\triangleq \! \big[B_{1,1},B_{1,2},B_{2}\big]$ and $\boldsymbol{\mathrm{p}} \!\triangleq \! \big[p_{1,1},p_{1,2},p_{2}\big]$, and (\ref{e14a}) defines the objective function and optimization variables; (\ref{e14b}) imposes a limit of $\xi_{q}^{th}$ on the SDVP of $x_{q}$, where $\xi_{q}^{th}\!$ is the SDVP threshold; (\ref{e14c}) indicates the constraint on the decoding error probability of $x_{q}$, with $\varepsilon_{q}^{th}$ as the decoding error probability threshold; (\ref{e14d}) restricts the transmit power $p_{u}$ of IIoT device-$u$ to a maximum value of $p_{max}$; (\ref{e14e}) specifies the constraint on the short-packet size of IIoT device-$u$. Finally, (\ref{e14f}) defines the range of values for the rate-splitting ratio.

\vspace{-1.4em}

\subsection{The SQP-driven Transmit Power Minimization Problem}

\par Given $\left(\xi_{q}^{th},w_{q}^{th},\varepsilon_{q}^{th}\right), q \! \in \! \mathcal{Q}$, the SQP-driven transmit power minimization problem can be formulated as

\begin{subequations}
   \begin{align}
      \mathcal{P}2:& \quad \min_{\{\boldsymbol{\mathrm{B}},\boldsymbol{\mathrm{p}},\alpha\}} \sum\limits_{q \in \mathcal{Q}} p_{q}, \label{e15a}\\
      \quad s.t. & \quad (\ref{e14b})-(\ref{e14f}).\label{e15b}
   \end{align}
\end{subequations}

\par The intertwined co-channel interference \cite{clerckx2023primer,mao2022rate} and the highly-varying wireless fading channels \cite{10382447,Yuang2023When} in RSMA-xURLLC-IIoT network architecture render the analytical expression for SDVP in (\ref{e14b}) is inaccessible \cite{fidler2014guide,chen2023streaming}. Additionally, the closed-form expression for the decoding error probability of short-packet data transmission becomes exceedingly complex in FBL regimes \cite{chen2023streaming}. As a result, $\mathcal{P}1$ and $\mathcal{P}2$ become exceedingly intractable and cannot be tackled analytically. Fortunately, by leveraging the proposed SNC-SQP theoretical framework, the inaccessible SDVP can be converted into the manageable UB-SDVP, which not only provides insightful theoretical guidance but also guarantees relatively conservative information bit transmission and power allocation schemes. Consequently, (\ref{e14b}) in $\mathcal{P}1$ and $\mathcal{P}2$ can be reformulated as follows:

\vspace{-1em}

\begin{subequations}\label{e16}
  \begin{align}
    & \mathcal{K}_{q}\left(w_{q}^{th},\boldsymbol{\mathrm{B}},\boldsymbol{\mathrm{p}},\alpha\right) = \inf_{0 < \theta_{q} \leq \theta_{0}}\!\!\left\{\!\!\frac{\left(\overline{\mathbb{M}}_{s_{q}}\!\!\left(\theta_{q}\right)\right)^{w_{q}^{th}}}{1\!-\!\mathbb{M}_{a_{q}}\!\!\left(\theta_{q}\right)\cdot\overline{\mathbb{M}}_{s_{q}}\!\!\left(\theta_{q}\right)}\!\!\right\} \! \leq \! \xi_{q}^{th},\\
    & \mathbb{M}_{a_{q}}\!\!\left(\theta_{q}\right)\cdot\overline{\mathbb{M}}_{s_{q}}\!\!\left(\theta_{q}\right) < 1,\ \forall q \in \mathcal{Q},
  \end{align}
\end{subequations}
where (\ref{e16}a) imposes a limit of $\xi_{q}^{th}$ on the UB-SDVP for stream $x_{q}$, and $(\ref{e16}b)$ represents the stability condition of UB-SDVP.

\vspace{-1em}

\subsection{Proposed Solutions}

\par Due to the complex constraints, solving $\mathcal{P}1$ and $\mathcal{P}2$ remains non-trivial. Therefore, we focus on exploring the intrinsic properties of $\mathcal{P}1$ and $\mathcal{P}2$.

\begin{corollary}\label{coro1}
   Given the transmit power $\boldsymbol{\mathrm{p}}$, the short-packet size $\boldsymbol{\mathrm{B}}$, and the rate-allocation ratio $\alpha$, both (\ref{e16}a) and (\ref{e16}b) are the convex functions with respect to $\theta_{q}$, and there exists the maximum feasible region $\left[0,\theta_{q}^{max}\left(\boldsymbol{\mathrm{p}},\boldsymbol{\mathrm{B}},\alpha\right)\right]$ and the optimal value of $\theta_{q}^{\ast}\!\left(\boldsymbol{\mathrm{p}},\boldsymbol{\mathrm{B}},\alpha\right)$ for $x_{q}, q\!\in\!\mathcal{Q}$.
\end{corollary}

\begin{proof}
   The detailed proof of \textbf{Corollary 1} can be found in Appendix C.
\end{proof}

\vspace{-0.5em}

\par From \textbf{Corollary \ref{coro1}}, $\left[0,\theta_{q}^{max}\!\left(\boldsymbol{\mathrm{p}},\boldsymbol{\mathrm{B}},\alpha\right)\right]$ can be easily determined by the one-dimensional search method. According to the properties of the convex function, we can easily determine $\theta_{q}^{\ast}\!\left(\boldsymbol{\mathrm{p}},\boldsymbol{\mathrm{B}},\alpha\right)$ through stochastic gradient descent (SGD) method\footnote{Due to space limitations, we have omitted detailed descriptions of one-dimensional search method and SGD method here since they are the fundamental content of optimization theory.}.

\vspace{-0.5em}
\begin{corollary}\label{coro2}
   The decoding error probability $\boldsymbol{\bar{\epsilon}}_{q}\!\left(\boldsymbol{\mathrm{B}},\boldsymbol{\mathrm{p}}\right)$ given by $(\ref{e7})$ is strictly decreasing with respect to $p_{q}$, and strictly increasing with respect to $B_{q}$.
\end{corollary}

\vspace{-1.3em}

\begin{proof}
 According to (\ref{e7}), we can obtain $\frac{\partial \boldsymbol{\bar{\epsilon}}_{q}\!\left(\boldsymbol{\mathrm{B}},\boldsymbol{\mathrm{p}}\right)}{\partial p_{q}}$ as follows:
 \vspace{-0.2em}
 \begin{equation*}
   \frac{\partial \boldsymbol{\bar{\epsilon}}_{q}\!\left(\boldsymbol{\mathrm{B}},\boldsymbol{\mathrm{p}}\right)}{\partial p_{q}} = \frac{\partial \widehat{\Gamma}_{q}}{\partial p_{q}} \cdot \frac{\partial \boldsymbol{\bar{\epsilon}}_{q}\!\left(\boldsymbol{\mathrm{B}},\boldsymbol{\mathrm{p}}\right)}{\partial \widehat{\Gamma}_{q}} = \frac{\partial \widehat{\Gamma}_{q}}{\partial p_{q}} \cdot \mathbb{E}_{\widehat{\Gamma}_{q}}\!\bigg[\frac{\partial Q\big(f(\boldsymbol{\mathrm{B}},\boldsymbol{\mathrm{p}})\big)}{\partial \widehat{\Gamma}_{q}}\bigg].
 \end{equation*}

  From \citep[Theorems 5]{10382447}, $\boldsymbol{\bar{\epsilon}}_{q}\left(\boldsymbol{\mathrm{B}},\boldsymbol{\mathrm{p}}\right)$ is a monotonically decreasing function with respect to $\widehat{\Gamma}_{q}$. Moreover, it can be easily obtained that $\frac{\partial \widehat{\Gamma}_{q}}{\partial p_{q}} > 0$, thus $\frac{\partial \boldsymbol{\bar{\epsilon}}_{q}\!\left(\boldsymbol{\mathrm{B}},\boldsymbol{\mathrm{p}}\right)}{\partial p_{q}} < 0$.

 Next, we prove the monotonic increasing property of $\boldsymbol{\bar{\epsilon}}_{q}\!\left(\boldsymbol{\mathrm{B}},\boldsymbol{\mathrm{p}}\right)$ with respect to $B_{q}$. The partial derivative of $\boldsymbol{\bar{\epsilon}}_{q}\!\left(\boldsymbol{\mathrm{B}},\boldsymbol{\mathrm{p}}\right)$ in $B_{q}$ can be given by

 \begin{equation}\label{e17}
    \frac{\partial \boldsymbol{\bar{\epsilon}}_{q}\!\left(\boldsymbol{\mathrm{B}},\boldsymbol{\mathrm{p}}\right)}{\partial B_{q}} = \mathbb{E}_{\widehat{\Gamma}_{q}}\!\bigg[\frac{\partial Q\big(f(\boldsymbol{\mathrm{B}},\boldsymbol{\mathrm{p}})\big)}{\partial B_{q}}\bigg].
 \end{equation}

 The partial derivative of $Q\big(f(\boldsymbol{\mathrm{B}},\boldsymbol{\mathrm{p}})\big)$ with respect to $B_{q}$ can be derived as follows:
 \begin{equation}\label{e18}
    \frac{\partial Q\big(f(\boldsymbol{\mathrm{B}},\boldsymbol{\mathrm{p}})\big)}{\partial B_{q}} = \frac{\ln 2 \exp\left(-\frac{1}{2}f^{2}(\boldsymbol{\mathrm{B}},\boldsymbol{\mathrm{p}})\right)}{\sqrt{2\pi N_{d} \left(1 - \frac{1}{(1 + \widehat{\Gamma})^{2}}\right)}} > 0.
    \vspace{-0.2em}
 \end{equation}

 Combining (\ref{e17}) and (\ref{e18}), it is obvious that $\boldsymbol{\bar{\epsilon}}_{q}\!\left(\boldsymbol{\mathrm{B}},\boldsymbol{\mathrm{p}}\right)$ is a strictly increasing function of $B_{q}$. So the proof of Corollary 2 is concluded.
\end{proof}

\par In addition, we also investigate the intrinsic properties of UB-SDVP as follows:

\begin{corollary}\label{coro3}
   The UB-SDVP $\mathcal{K}_{q}\left(w_{q}^{th},\boldsymbol{\mathrm{B}},\boldsymbol{\mathrm{p}},\alpha\right)$ given by (\ref{e16}a) is strictly decreasing with respect to transmit power $p_{q}$, and strictly increasing with respect to the information bit number (i.e., the short-packet size) $B_{q}$.
\end{corollary}

\begin{proof}
   The partial derivative of $\mathcal{K}_{q}\left(w_{q}^{th},\boldsymbol{\mathrm{B}},\boldsymbol{\mathrm{p}},\alpha\right)$ in $p_{q}$ can be given as follows:
   \begin{equation}\label{e19}
      \frac{\partial \mathcal{K}_{q}\left(w_{q}^{th},\boldsymbol{\mathrm{B}},\boldsymbol{\mathrm{p}},\alpha\right)}{\partial p_{q}} = \frac{\partial \mathcal{K}_{q}\left(w_{q}^{th},\boldsymbol{\mathrm{B}},\boldsymbol{\mathrm{p}},\alpha\right)}{\partial \overline{\mathbb{M}}_{s_{q}}\left(\theta_{q}\right)}\frac{\partial \overline{\mathbb{M}}_{s_{q}}\left(\theta_{q}\right)}{\partial p_{q}}.
   \end{equation}

   It can be easily obtained that $\frac{\partial \mathcal{K}_{q}\left(w_{q}^{th},\boldsymbol{\mathrm{B}},\boldsymbol{\mathrm{p}},\alpha\right)}{\partial \overline{\mathbb{M}}_{s_{q}}\left(\theta_{q}\right)} > 0$. Next, we examine the monotonicity of $\overline{\mathbb{M}}_{s_{q}}\left(\theta_{q}\right)$ as follows:

   \begin{small}
   \begin{equation}\label{e20}
      \frac{\partial \overline{\mathbb{M}}_{s_{q}}\left(\theta_{q}\right)}{\partial p_{q}} = 0 + \big(1\!-\!\mathbb{E}_{\widehat{\Gamma}_{\!q}}\!\!\left[e^{-\theta_{q} N_{d} R_{q}}\right]\!\big)\!\cdot\!\frac{\partial}{\partial p_{q}}\big\{\boldsymbol{\bar{\epsilon}}\left(\boldsymbol{\mathrm{B}},\boldsymbol{\mathrm{p}}\right)\!\big\} < 0.
      \vspace{-0.5em}
   \end{equation}
   \end{small}
   Thus, $\mathcal{K}_{q}\!\!\left(w_{q}^{th},\boldsymbol{\mathrm{B}},\boldsymbol{\mathrm{p}},\alpha\right)$ is strictly decreasing with $p_{q}$. And the partial derivative of $\mathcal{K}_{q}\!\!\left(w_{q}^{th},\boldsymbol{\mathrm{B}},\boldsymbol{\mathrm{p}},\alpha\right)$ in $B_{q}$ can be given as follows:
   \vspace{-0.6em}
   \begin{equation}\label{e21}
      \frac{\partial \mathcal{K}_{q}\!\left(w_{q}^{th},\boldsymbol{\mathrm{B}},\boldsymbol{\mathrm{p}},\alpha\right)}{\partial B_{q}} = \frac{\partial \mathcal{K}_{q}\!\left(w_{q}^{th},\boldsymbol{\mathrm{B}},\boldsymbol{\mathrm{p}},\alpha\right)}{\partial \overline{\mathbb{M}}_{s_{q}}\left(\theta_{q}\right)}\frac{\partial \overline{\mathbb{M}}_{s_{q}}\left(\theta_{q}\right)}{\partial B_{q}}.
   \end{equation}

   Then, we examine the notation of $\frac{\partial \overline{\mathbb{M}}_{s_{q}}\left(\theta_{q}\right)}{\partial B_{q}}$ as follows:
   \vspace{-0.6em}
   \begin{equation}\label{e22}
     \begin{aligned}
      & \frac{\partial \overline{\mathbb{M}}_{s_{q}}\left(\theta_{q}\right)}{\partial B_{q}} = \frac{\partial }{\partial B_{q}}\!\!\left\{\!\mathbb{E}_{\widehat{\Gamma}_{\!q}}\!\!\left[e^{-\theta_{q} B_{q}}\right]\!\right\} + \frac{\partial }{\partial B_{q}}\big\{\!\boldsymbol{\bar{\epsilon}}\left(\boldsymbol{\mathrm{B}},\boldsymbol{\mathrm{p}}\right)\!\!\big\} - \\
      & \left(\!\boldsymbol{\bar{\epsilon}}\left(\boldsymbol{\mathrm{B}},\boldsymbol{\mathrm{p}}\right) \frac{\partial}{\partial B_{q}}\!\!\left\{\!\mathbb{E}_{\widehat{\Gamma}_{\!q}}\!\!\left[e^{-\theta_{q} B_{q}}\right]\!\right\} + \mathbb{E}_{\widehat{\Gamma}_{\!q}}\!\!\left[e^{-\theta_{q} B_{q}}\right]\frac{\partial }{\partial B_{q}}\!\big\{\!\boldsymbol{\bar{\epsilon}}\left(\boldsymbol{\mathrm{B}},\boldsymbol{\mathrm{p}}\right)\!\!\big\}\!\!\right)\\
      & = \!\big(1\!-\! \boldsymbol{\bar{\epsilon}}\left(\boldsymbol{\mathrm{B}},\boldsymbol{\mathrm{p}}\right)\big)\mathbb{E}_{\widehat{\Gamma}_{\!q}}\!\!\left[e^{-\theta_{q} B_{q}}\right] \!\!+\!\! \left(1 \!-\! e^{-\theta_{q}B_{q}}\!\right)\!\frac{\partial }{\partial B_{q}}\big\{\!\boldsymbol{\bar{\epsilon}}\left(\boldsymbol{\mathrm{B}},\boldsymbol{\mathrm{p}}\right)\!\!\big\} \!>\! 0.
     \end{aligned}
   \end{equation}
   From (\ref{e21}) and (\ref{e22}), we can obtain that $\mathcal{K}_{q}\!\!\left(w_{q}^{th},\boldsymbol{\mathrm{B}},\boldsymbol{\mathrm{p}},\alpha\right)$ is a strictly increasing function of $B_{q}$.
\end{proof}

\par From \textbf{Corollary 2} and \textbf{Corollary 3}, \textbf{Theorem 3} can be derived as follows:

\begin{theorem}\label{theo3}
  For $\mathcal{P}1$, the optimal transmit power is determined as $p_{q}^{\star} \!\!=\!\! p_{max}$, and the optimal rate-allocation ratio is $\alpha^{\star} \!=\! \frac{B_{1,1}^{\star}}{B_{1,1}^{\star} + B_{1,2}^{\star}}$, where $B_{1,1}^{\star}$ and $B_{1,2}^{\star}$ represent the optimal short-packet sizes for $x_{1,1}$ and $x_{1,2}$, respectively. As for $\mathcal{P}2$, the optimal short-packet size is determined as $B_{q}^{\star} \!=\! B_{min}, q \!\in\! \mathcal{Q}$, and the optimal rate-splitting ratio is $\alpha^{\star} \!=\! \frac{1}{2}$.
\end{theorem}

\vspace{-0.5em}

\begin{proof}
    According to \textbf{Corollary 2} and \textbf{Corollary 3}, the transmit power should be as high as possible to maximize the short-packet size. Therefore, we have $p_{q}^{\star} \!\!=\!\! p_{max}, q \! \in \! \mathcal{Q}$. We denote that the optimal short-packet sizes for $x_{1,1}$ and $x_{1,2}$ are $B_{1,1}^{\star}$ and $B_{1,2}^{\star}$, respectively. Based on (\ref{e7}), we have $\frac{R_{1,1}^{\star}}{R_{1,2}^{\star}} = \frac{B_{1,1}^{\star}/N_{d}}{B_{1,2}^{\star}/N_{d}} = \frac{\alpha^{\star}}{1 - \alpha^{\star}}$, thus, $\alpha^{\star} = \frac{B_{1,1}^{\star}}{B_{1,1}^{\star} + B_{1,2}^{\star}}$. Furthermore, we have $B_{q}^{\star} = B_{min}$ to minimize the transmit power of short-packet communications, and we can obtain that $\alpha^{\star} = \frac{1}{2}$ since $B_{1,1}^{\star} = B_{1,2}^{\star}$.
\end{proof}

\vspace{-1em}

\begin{algorithm}[h]
\small
\setstretch{0.35}
\label{algorithm1}
\caption{The proposed TSSO algorithm.}
\LinesNumbered
\KwIn{$(w_{q}^{th},\xi_{q}^{th})$; $p_{max}$; $B_{min}$; $B_{max}$; convergence criteria $\pi_{th}$; maximum iterations $M_{iter}$.}
\KwOut{The optimal solution for $\mathcal{P}1$: $\{\boldsymbol{\mathrm{\widetilde{B}}},\boldsymbol{\mathrm{\widetilde{p}}},\widetilde{\alpha}\}$; The optimal solution for $\mathcal{P}2$: $\{\boldsymbol{\mathrm{\widehat{B}}},\boldsymbol{\mathrm{\widehat{p}}},\widehat{\alpha}\}$.}
\textcolor[rgb]{1.00,0.00,0.50}{\tcp{\textbf{Solving $\mathcal{P}1$}}}
Setting $\widetilde{p}_{q} = p_{max}$, lower-search bound $B_{q}^{l} = B_{min}$, upper-search bound $B_{q}^{u} = B_{max}$, iteration index $k_{q} = 1$, UB-SDVP $\xi_{q} = 0$, decoding error probability $\varepsilon_{q} = 0$;\\
\textcolor[rgb]{1.00,0.50,0.75}{\tcp{\!\!\!\!\textbf{Step} \!\!\!\!\!\! \textbf{1:}\!\!\! Optimize $\!\!\widetilde{B}_{2}\!\!$ for stream $\!x_{2}$.}}
\While{$|\xi_{2}/\xi_{2}^{th} \!-\! 1| \!>\! \pi_{th}$ $\boldsymbol{\mathrm{or}}$ $|\varepsilon_{2}/\varepsilon_{2}^{th} \!-\! 1| \!>\! \pi_{th}$ $\boldsymbol{\mathrm{and}}$ $k_{2} \!<\! M_{iter}$}{
$\widetilde{B}_{2} \leftarrow (B_{2}^{l} + B_{2}^{u})/2$;\\
Determine feasible region $\left[0,\theta_{2}^{max}\left(\boldsymbol{\mathrm{\widetilde{p}}},\widetilde{B}_{2}\right)\right]$;\\
Determine optimal QoS exponent $\theta_{2}^{\ast}\!\left(\boldsymbol{\mathrm{\widetilde{p}}},\widetilde{B}_{2}\right)$;\\
Update $\xi_{2} = \mathcal{K}_{2}\big(w_{2}^{th},\widetilde{B}_{2},\boldsymbol{\mathrm{\widetilde{p}}}\big)$ and $\varepsilon_{2} = \boldsymbol{\bar{\epsilon}}\left(\widetilde{B}_{2},\boldsymbol{\mathrm{\widetilde{p}}}\right)$;\\
\eIf{$\xi_{2}/\xi_{2}^{th} > 1$ $\boldsymbol{\mathrm{or}}$ $\varepsilon_{2}/\varepsilon_{2}^{th} > 1$}{
$B_{2}^{l} \leftarrow (B_{2}^{l} + B_{2}^{u})/2$;\\
}
{
$B_{2}^{u} \leftarrow (B_{2}^{l} + B_{2}^{u})/2$;\\
}
$k_{2} \leftarrow k_{2} + 1$;\\
}
\textcolor[rgb]{1.00,0.50,0.75}{\tcp{\!\!\!\!\textbf{Step} \!\!\!\!\!\! \textbf{2:}\!\!\! Optimize $\widetilde{B}_{1,2}$ for stream $x_{1,2}$.}}
Use $\widetilde{B}_{2}$ as the input of \textbf{Step 2};\\
Perform similar procedures of \textbf{Step 1:} \textbf{4-15} to solve $\widetilde{B}_{1,2}$;\\
\textcolor[rgb]{1.00,0.50,0.75}{\tcp{\!\!\!\!\textbf{Step} \!\!\!\!\!\! \textbf{3:}\!\!\! Optimize $B_{1,1}$ for stream $x_{1,1}$.}}
Use $\widetilde{B}_{2}$ and $\widetilde{B}_{1,2}$ as the inputs of \textbf{Step 3};\\
Perform similar procedures of \textbf{Step 1: 4-15} to solve $\widetilde{B}_{1,1}$;\\
Setting the optimal rate-split ratio $\widetilde{\alpha} = \widetilde{B}_{1,1}/\left(\widetilde{B}_{1,1} + \widetilde{B}_{1,2}\right)$;\\
\textcolor[rgb]{1.00,0.00,0.50}{\tcp{\textbf{Solving $\mathcal{P}2$}}}
Setting $\widehat{B}_{q} = B_{min}$, $\widehat{\alpha} = 1/2$, lower-search bound $p_{q}^{l} = 0$, upper-search bound $p_{q}^{u} = p_{max}$, $k_{q} = 1$, $\xi_{q} = 0$, $\varepsilon_{q} = 0$;\\
\textcolor[rgb]{1.00,0.50,0.75}{\tcp{\!\!\!\!\textbf{Step} \!\!\!\!\!\! \textbf{1:}\!\!\! Optimize $\!\!\widehat{p}_{2}\!\!$ for stream $\!x_{2}$.}}
\While{$|\xi_{2}/\xi_{2}^{th} \!-\! 1| \!>\! \pi_{th}$ $\boldsymbol{\mathrm{or}}$ $|\varepsilon_{2}/\varepsilon_{2}^{th} \!-\! 1| \!>\! \pi_{th}$ $\boldsymbol{\mathrm{and}}$ $k_{2} \!<\! M_{iter}$}{
$\widehat{p}_{2} \leftarrow (p_{2}^{l} + p_{2}^{u})/2$;\\
Determine feasible region $\left[0,\theta_{2}^{max}\left(\widehat{p}_{2},\boldsymbol{\mathrm{\widehat{B}}}\right)\right]$;\\
Determine optimal QoS exponent $\theta_{2}^{\ast}\!\left(\widehat{p}_{2},\boldsymbol{\mathrm{\widehat{B}}}\right)$;\\
Update $\xi_{2} = \mathcal{K}_{2}\big(w_{2}^{th},\widetilde{B}_{2},\boldsymbol{\mathrm{\widetilde{p}}}\big)$ and $\varepsilon_{2} = \boldsymbol{\bar{\epsilon}}\left(\widetilde{B}_{2},\boldsymbol{\mathrm{\widetilde{p}}}\right)$;\\
\eIf{$\xi_{2}/\xi_{2}^{th} > 1$ $\boldsymbol{\mathrm{or}}$ $\varepsilon_{2}/\varepsilon_{2}^{th} > 1$}{
$p_{2}^{l} \leftarrow (p_{2}^{l} + p_{2}^{u})/2$;\\
}
{
$p_{2}^{u} \leftarrow (p_{2}^{l} + p_{2}^{u})/2$;\\
}
$k_{2} \leftarrow k_{2} + 1$;\\
}
\textcolor[rgb]{1.00,0.50,0.75}{\tcp{\!\!\!\!\textbf{Step} \!\!\!\!\!\! \textbf{2:}\!\!\! Optimize $\widehat{p}_{1,2}$ for stream $x_{1,2}$.}}
Use $\widetilde{B}_{2}$ as the input of \textbf{Step 2};\\
Perform similar procedures of \textbf{Step 1:} \textbf{26-37} to solve $\widehat{p}_{1,2}$;\\
\textcolor[rgb]{1.00,0.50,0.75}{\tcp{\!\!\!\!\textbf{Step} \!\!\!\!\!\! \textbf{3:}\!\!\! Optimize $\widehat{p}_{1,1}$ for stream $x_{1,1}$.}}
Use $\widehat{p}_{2}$ and $\widehat{p}_{1,2}$ as the inputs of \textbf{Step 3};\\
Perform similar procedures of \textbf{Step 1: 26-37} to solve $\widehat{p}_{1,1}$;
\end{algorithm}

\begin{figure*}[t]
  \centering
  \begin{minipage}[b]{0.45\textwidth}
    \includegraphics[scale=0.028]{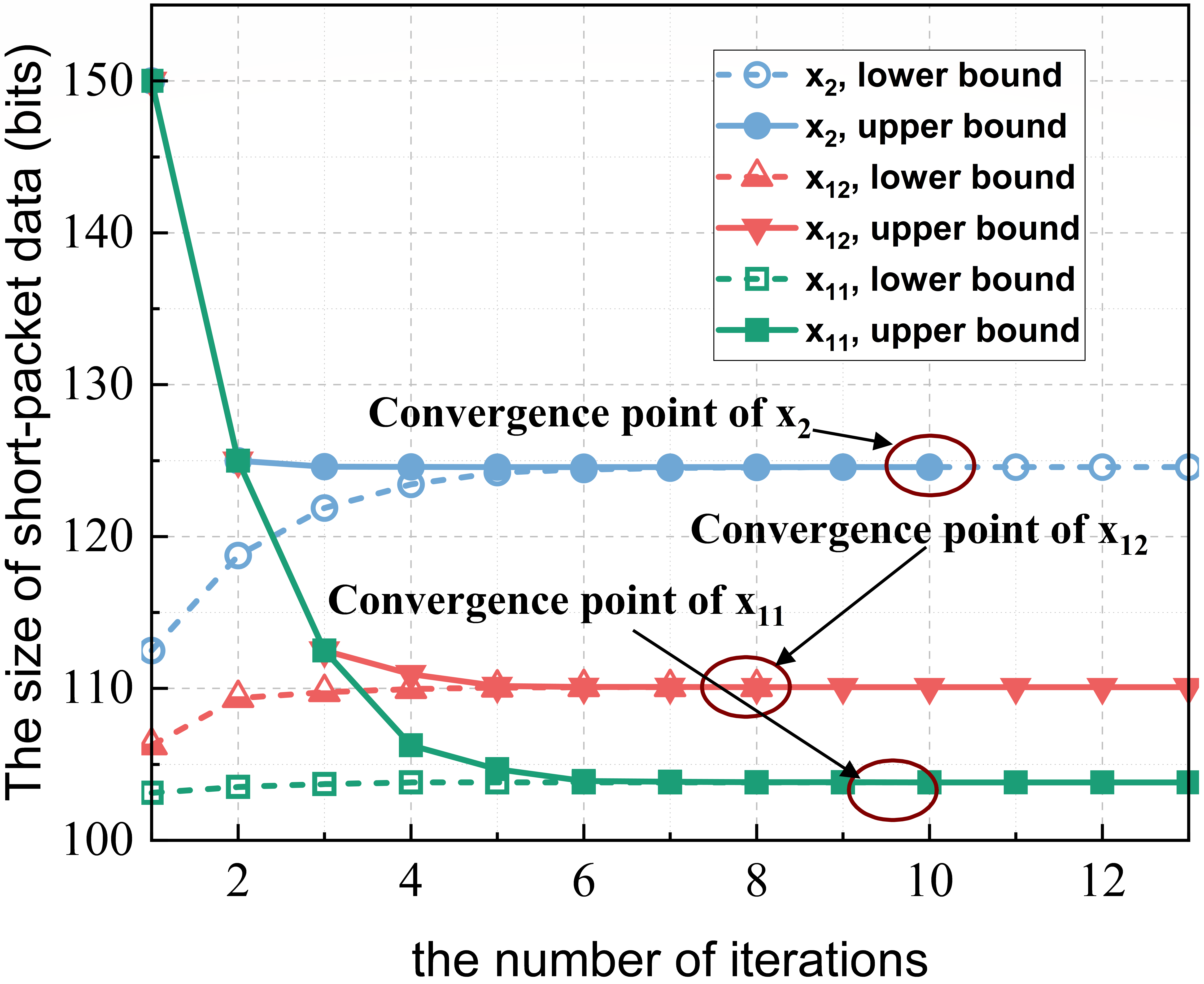}
    \caption{Convergence behavior of short-packet size maximization in the TSSO algorithm.}
    \label{fig_2}
  \end{minipage}
  \vspace{-0.8em}
  \hfill 
  \begin{minipage}[b]{0.45\textwidth}
    \includegraphics[scale=0.029]{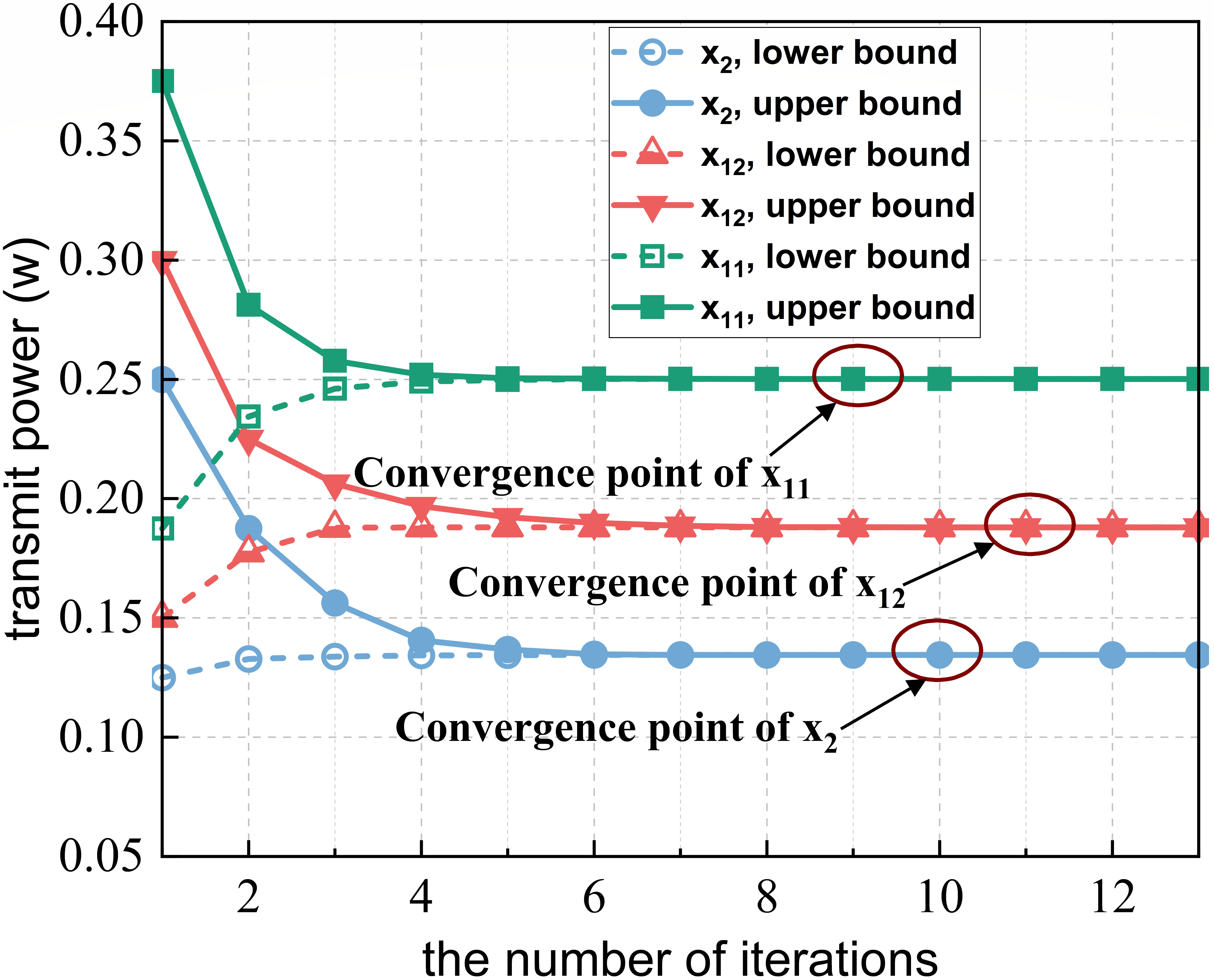}
    \caption{Convergence behavior of transmit power minimization in the TSSO algorithm.}
    \label{fig_3}
  \end{minipage}
  \vspace{-0.8em}
\end{figure*}

\par Following RSMA principles \cite{clerckx2023primer,mao2022rate}, $x_{2}$ is immune to interference, $x_{1,2}$ experiences co-channel interference from $x_{2}$, and $x_{1,1}$ is affected by co-channel interference from both $x_{2}$ and $x_{1,2}$. Based on the distinct characteristics of streams, we propose leveraging sequential optimization techniques to effectively solve $\mathcal{P}1$ and $\mathcal{P}2$ \cite{du2004sequential}. Specifically, we decompose $\mathcal{P}1$ and $\mathcal{P}2$ into three subproblems that need to be tackled sequentially. For \textbf{Subproblem I}, we focus on maximizing(or minimizing) the short-packet size (or transmit power) of stream $x_{2}$. For \textbf{Subproblem II} and \textbf{Subproblem III}, we concentrate on maximizing (or minimizing) the short-packet size (or transmit power) of streams $x_{1,2}$ and $x_{1,1}$, respectively.

\par It can be easily observed that if these three subproblems are optimized sequentially in the reverse decoding order, i.e., $x_{2} \!\rightarrow \!x_{1,2} \!\rightarrow\! x_{1,1}$, $\mathcal{P}1$ and $\mathcal{P}2$ can be successfully decoupled into three independent subproblems \cite{du2004sequential}. For each subproblem, if $\mathcal{K}_{q}\!\left(w_{q}^{th},\boldsymbol{\mathrm{B}},\boldsymbol{\mathrm{p}},\alpha\right) \!\!\leq\!\! \xi_{q}^{th}$ and $\boldsymbol{\bar{\epsilon}}\left(\boldsymbol{\mathrm{B}},\boldsymbol{\mathrm{p}}\right) \!\!\leq\!\! \varepsilon_{q}^{th}$, it indicates that the short-packet size $B_{q}$ (or transmit power $p_{q}$) is too small (or large) and needs to be increase (or reduced); if $\mathcal{K}_{q}\!\left(w_{q}^{th},\boldsymbol{\mathrm{B}},\boldsymbol{\mathrm{p}},\alpha\right) > \xi_{q}^{th}$ or $\boldsymbol{\bar{\epsilon}}_{q}\left(\boldsymbol{\mathrm{B}},\boldsymbol{\mathrm{p}}\right) > \varepsilon_{q}^{th}$, it indicates that the short-packet size $B_{q}$ (or $p_{q}$) is too large (or small) and needs to be reduced (or increased). In this case, we propose a low-complexity algorithm termed the three-step sequential optimization algorithm (TSSO) to effectively tackle $\mathcal{P}1$ and $\mathcal{P}2$. The detailed descriptions of TSSO are outlined in \textbf{Algorithm 1}.

\vspace{-0.5em}

\subsection{Computational Complexity Analysis}

\par \textbf{Algorithm 1} primarily involves the one-dimensional search method and SGD algorithm. The complexity of the one-dimensional search method can be represented as $\mathcal{O}\!\left(\log\!\left(\psi_{\theta}/\psi_{th}\!\right)\right)$, where $\psi_{\theta}$ and $\psi_{th}$ denote the step-length factor and step-length threshold, respectively. The complexity of the SGD method can be denoted as $\mathcal{O}\left(\log(|\nabla \mathcal{K}_{q}\left(w_{q}^{th},\boldsymbol{p},\alpha\right)\Lambda_{s}|/\phi_{th})\right)$, where $\Lambda_{s}$ and $\phi_{th}$ indicate the step-length factor and convergence accuracy, respectively. As a result, the complexity of \textbf{Algorithm 1} is $\mathcal{O}\!\left( \sum_{q\in\mathcal{Q}} k_{q} \big(\!\log\left(\psi_{\theta}/\psi_{th}\right) + \log(|\nabla\mathcal{K}_{q}\!\!\left(w_{q}^{th},\boldsymbol{p},\alpha\right)\Lambda_{s}|/\phi_{th})\big)\!\!\right)$, where $k_{q}$ is the number of iterations for \textbf{Step \emph{i}}, $i \in\left\{1,2,3\right\}$.

\section{Performance Evaluation}

\par In this section, extensive simulations are conducted to demonstrate the effectiveness of the proposed RSMA-xURLLC. The BS service radius is $500$ m, with each subcarrier operating at a bandwidth of $2$ MHz and a time-slot length of $0.5$ ms. Thus, the maximum available blocklength is $N_{0} \!=\! 10^{3}$ CUs. The lengths of the orthogonal pilot sequence is set to $N_{p,1} = N_{p,2} = 50$ CUs. The average arrival rate is $250$ kbps. The maximum uplink transmit power for each stream of IIoT devices is $1W$, the minimum short-packet size is $B_{min} = 80$ bits, and the maximum short-packet size is $B_{max} = 500$ bits. The noise power level is at $\!-176\!$ dBm/Hz, and the path-loss factor is $\!2.5$. Shadow fading is modeled as a lognormal distribution with a standard variance of $8$ dB and small-scale fading follows a Rayleigh fading distribution. The parameters involved in \textbf{Algorithm 1} are as follows: $\pi_{th} \!=\! \phi_{th} \!=\! 10^{-15}$, $M_{iter} \!=\! 200$, $\psi_{\theta} \!=\! 1$, $\psi_{th}\!=\!\Lambda_{s} \!=\! 10^{-6}$.

\vspace{-1em}

\subsection{Convergence Analysis}

\par In Fig. \ref{fig_2} and Fig. \ref{fig_3}, we showcase the superior convergence performance exhibited by the proposed TSSO in tackling $\mathcal{P}1$ and $\mathcal{P}2$, respectively. Fig. \ref{fig_2} reveals that the upper and lower bounds associated with short-packet sizes of streams $x_{2}$, $x_{12}$, and $x_{11}$ rapidly converge as the number of iterations increases. Specifically, the short-packet sizes of streams $x_{2}$, $x_{12}$, and $x_{11}$ converge after $10$, $8$, and $7$ iterations, respectively. Similarly, it can be observed from Fig. \ref{fig_3} that the upper and lower bounds of the required transmit power for streams $x_{2}$, $x_{12}$, and $x_{11}$ also converge rapidly with increasing iterations. In particular, the transmit power of streams $x_{2}$, $x_{12}$, and $x_{11}$ converge after $10$, $11$, and $9$ iterations, respectively. Henceforth, the proposed TSSO algorithm demonstrates exceptional convergence performance, culminating in stable solutions within several iterations. Significantly, stream $x_{11}$ requires the highest transmit power but it can only accommodate the smallest short-packet size. Following closely is stream $x_{12}$, while stream $x_{2}$ demands the lowest transmit power while supporting the largest short-packet data size. These distinctions primarily stem from the distinctive co-channel interference characteristics inherent to streams $x_{2}$, $x_{1,2}$, and $x_{1,2}$ themselves. Notably, $x_{11}$ experiences interference from both streams $x_{2}$ and $x_{12}$ simultaneously.

\vspace{-1em}

\subsection{Validation of SNC-SQP Theoretical Framework}

\par As illustrated in Fig. \ref{fig_4}, we substantiate the dependability of the developed SNC-SQP theoretical framework for capturing the SQP performance. We employ Monte Carlo methods to stochastically generate $10^{8}$ channels, meticulously evaluating the actual SDVP, yielding Sim-SDVP. Subsequently, we conduct a comprehensive comparative analysis by juxtaposing the UB-SDVP with the Sim-SDVP. The precisely calculated slopes of Sim-SDVP and UB-SDVP on the $X$-$log(Y)$ axis for $x_{11}$, $x_{12}$, and $x_{2}$ reveal near-identical slopes within the logarithmic domain. This compellingly demonstrates the accuracy of the developed SNC-based SQP theoretical framework in capturing the SQP performance of our developed RSMA-xURLLC-IIoT network architecture. Moreover, the observable horizontal gap between UB-SDVP and Sim-SDVP primarily arises from the fact that our theoretical framework is grounded in $(min,\times)$-algebra theory \cite{al2014network, fidler2014guide, fidler2010survey}, which can effectively utilize arrival processes and the distribution of fading channels to elucidate the queuing system of our developed RSMA-xURLLC-IIoT network architecture. These numerical results undeniably validate the feasibility of our theoretical framework in converting unavailable SDVP into manageable UB-SDVP, providing profound theoretical guidance for SQP-driven optimization problems. Additionally, the developed RSMA-xURLLC-IIoT network architecture incorporates the statistical QoS provisioning mechanism, which provides a more flexible and QoS-guaranteed short-packet size maximization scheme and transmit power allocation scheme, in contrast to traditional deterministic QoS provisioning mechanisms.

\vspace{-0.5em}

\begin{figure}[h]
\centering
\includegraphics[scale=0.028]{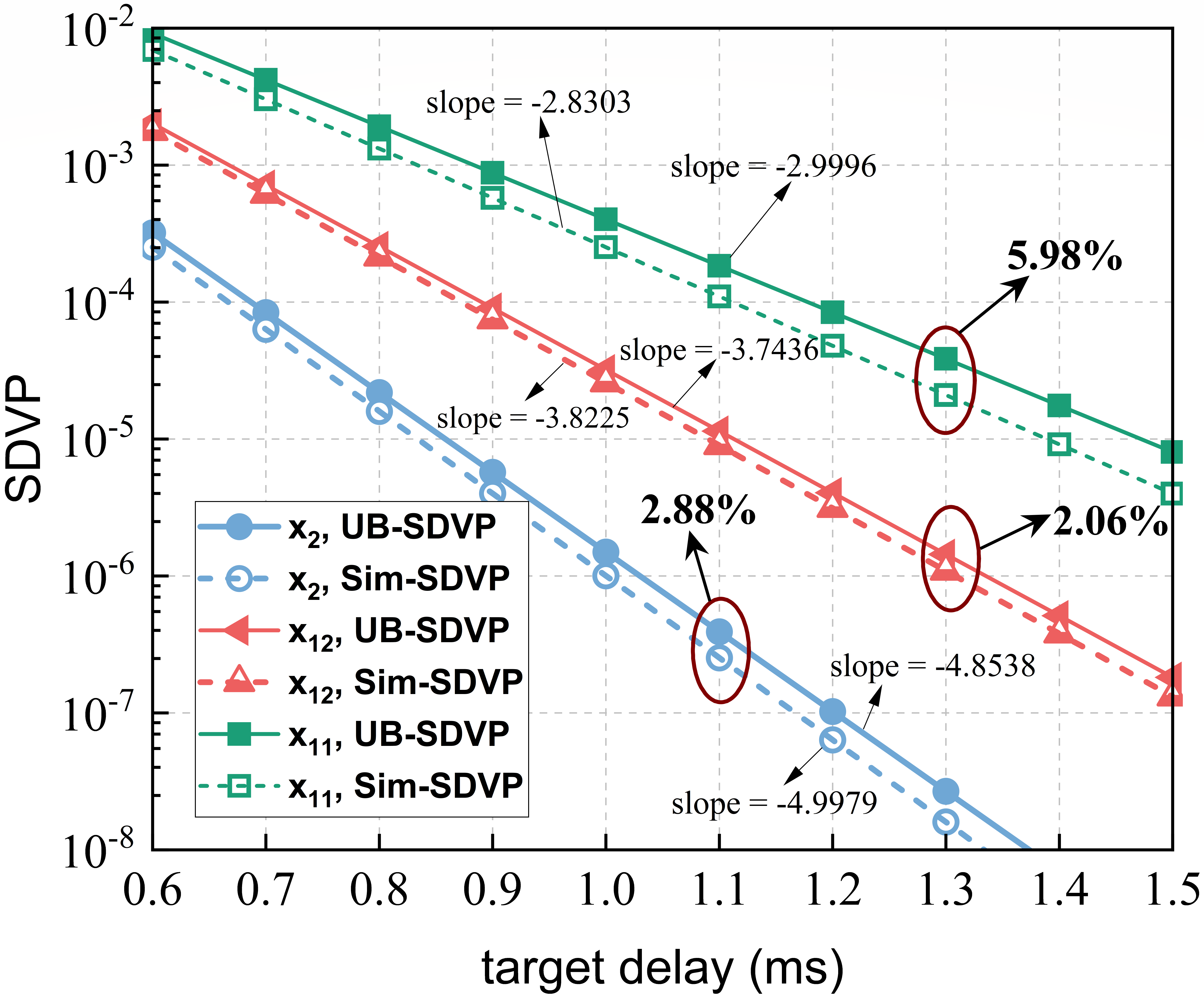}
\caption{Validation of the dependability for the developed SNC-SQP theoretical framework.}
\label{fig_4}
\end{figure}

\vspace{-2em}

\subsection{Comprehensive Performance Comparison}

\par To comprehensively evaluate the effectiveness of our developed RSMA-xURLLC-IIoT network architecture, we conduct an exhaustive performance comparison with state-of-the-art multi-access techniques exploited in current IIoT networks, including NOMA and OMA \cite{Yuang2023When,clerckx2023primer,mao2022rate}.

\par As shown in Fig. \ref{fig_5}, the relationship between the maximum short-packet size and target delay is investigated. We have carefully examined the sequential optimization process of streams $x_{2}$, $x_{12}$, and $x_{11}$ conducted by our proposed TSSO algorithm, the performance of the proposed schemes across varied scenarios, and comparisons with NOMA and OMA. It can be easily seen that as the target delay increases, the maximum short-packet size supported by the developed RSMA-xURLLC-IIoT network architecture improves as the target delay becomes relaxed. Moreover, it can be observed that increasing the allocated blocklength $N_{d}$ or relaxing the SDVP threshold $\xi_{th}$ contributes to enhancing the maximum short-packet size. With $N_{d} = 400$ CUs and $\xi_{th} = 10^{-6}$, the developed RSMA-xURLLC-IIoT network architecture outperforms NOMA and OMA schemes by margins of $21.8 \%$ and $30.2 \%$, respectively.

\vspace{-0.2em}

\begin{figure}[h]
\centering
\includegraphics[scale=0.028]{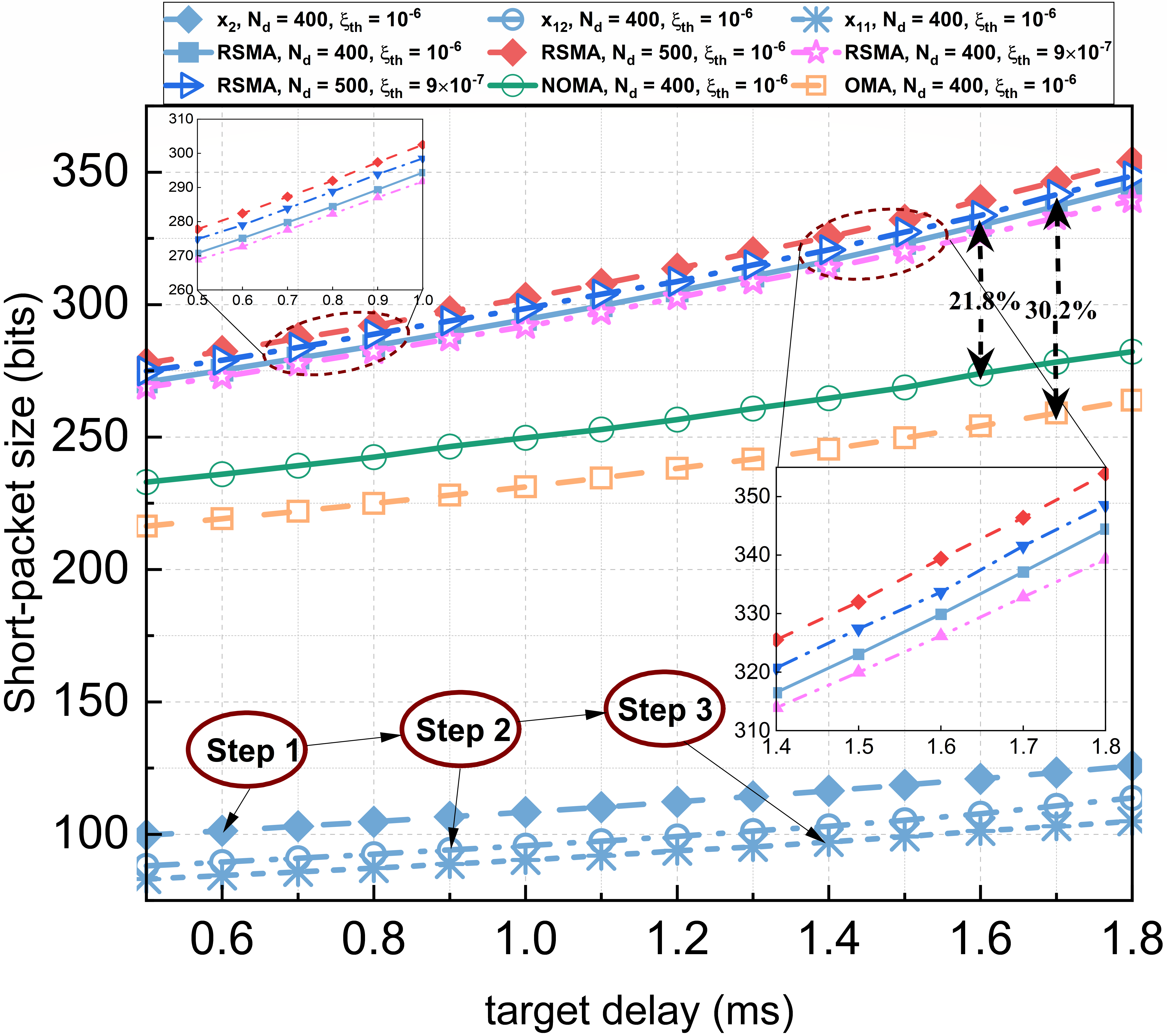}
\caption{The maximum short-packet size versus target delay.}
\label{fig_5}
\end{figure}

\vspace{-0.3em}

\par As depicted in Fig. \ref{fig_6}, the relationship between minimum transmit power and target latency is explored. It is readily apparent that when the target delay is relatively relaxed, the minimum transmit power required to guarantee the expected QoS requirements rapidly decreases. It demonstrates that the minimum required transmit power of our network architecture is sensitive to the target delay. Furthermore, increasing the allocated blocklength and relaxing the SDVP threshold effectively reduce the required minimum transmit power. Under $N_{d} = 400$ CUs and $\xi_{th} \!=\! 9\! \times \! 10^{-7}$, compared to NOMA and OMA schemes, the developed RSMA-xURLLC-IIoT network architecture achieves performance improvements in terms of power consumption of $44.15\%$ and $62.32\%$, respectively.

\vspace{-0.2em}

\begin{figure}[h]
\centering
\includegraphics[scale=0.028]{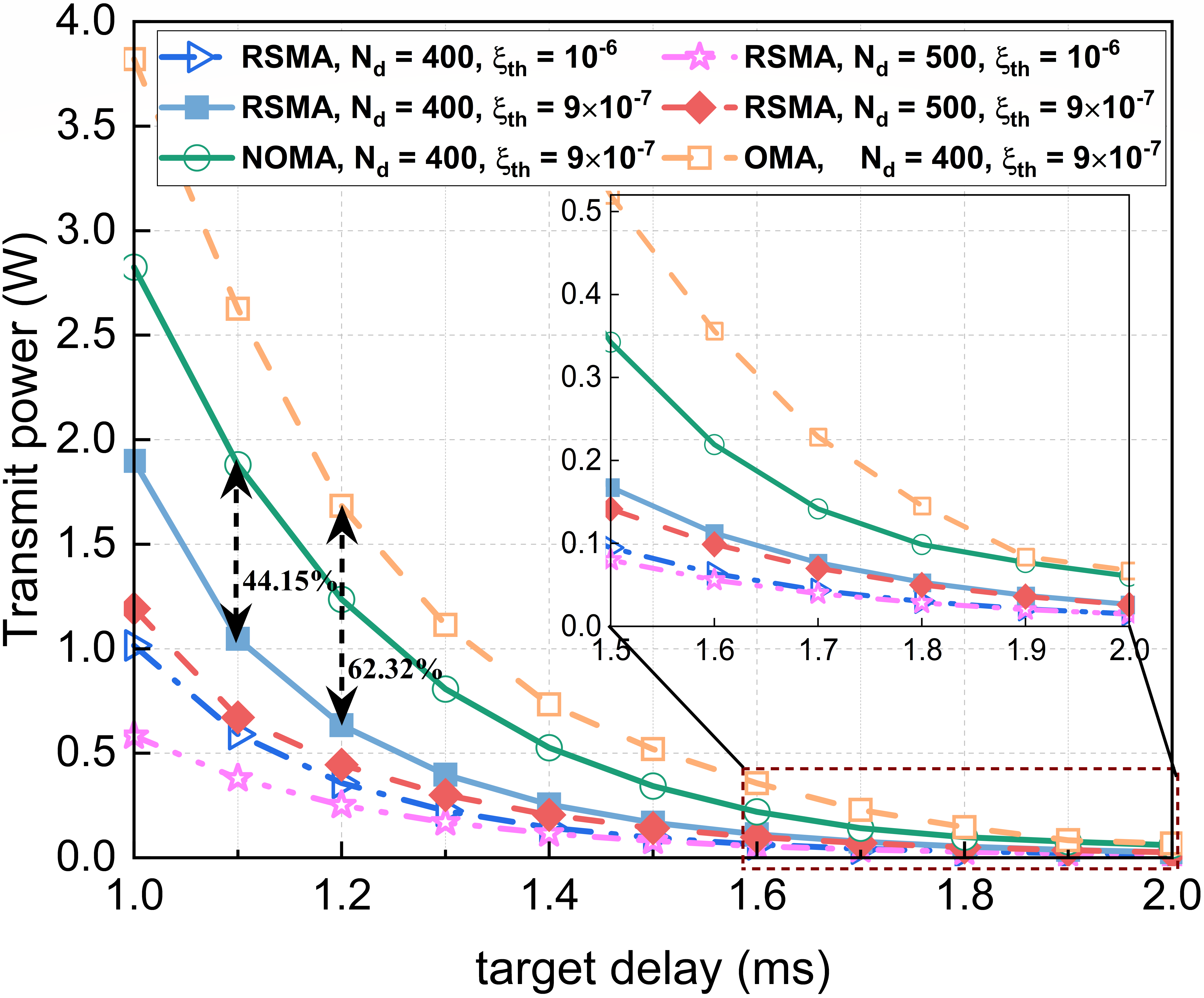}
\caption{The minimum transmit power versus target delay.}
\label{fig_6}
\end{figure}


\begin{figure*}[t]\label{fig_7}
\centering
  \subfigure[]{
   \includegraphics[scale=0.028]{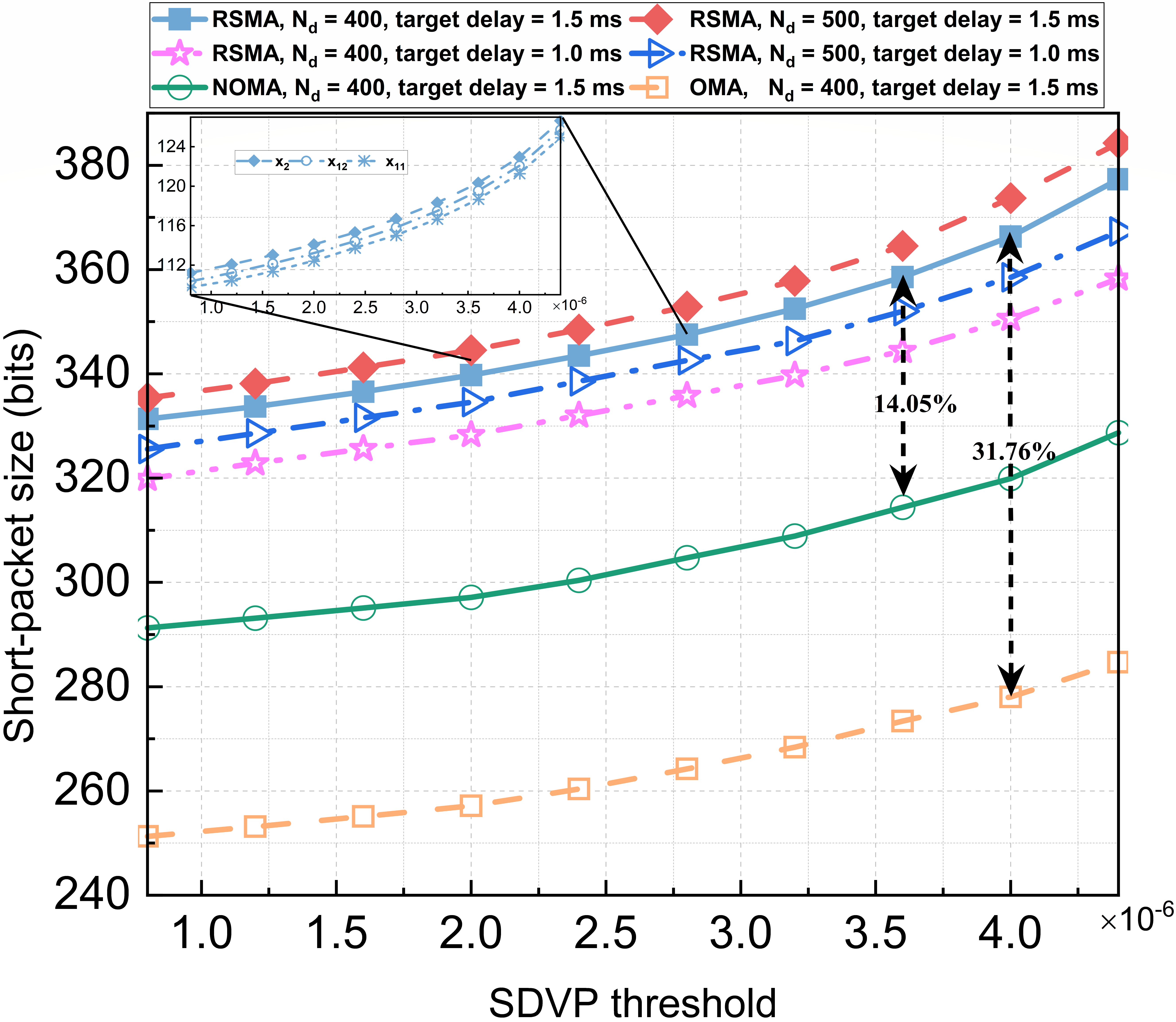}
  }
  \subfigure[]{
   \includegraphics[scale=0.028]{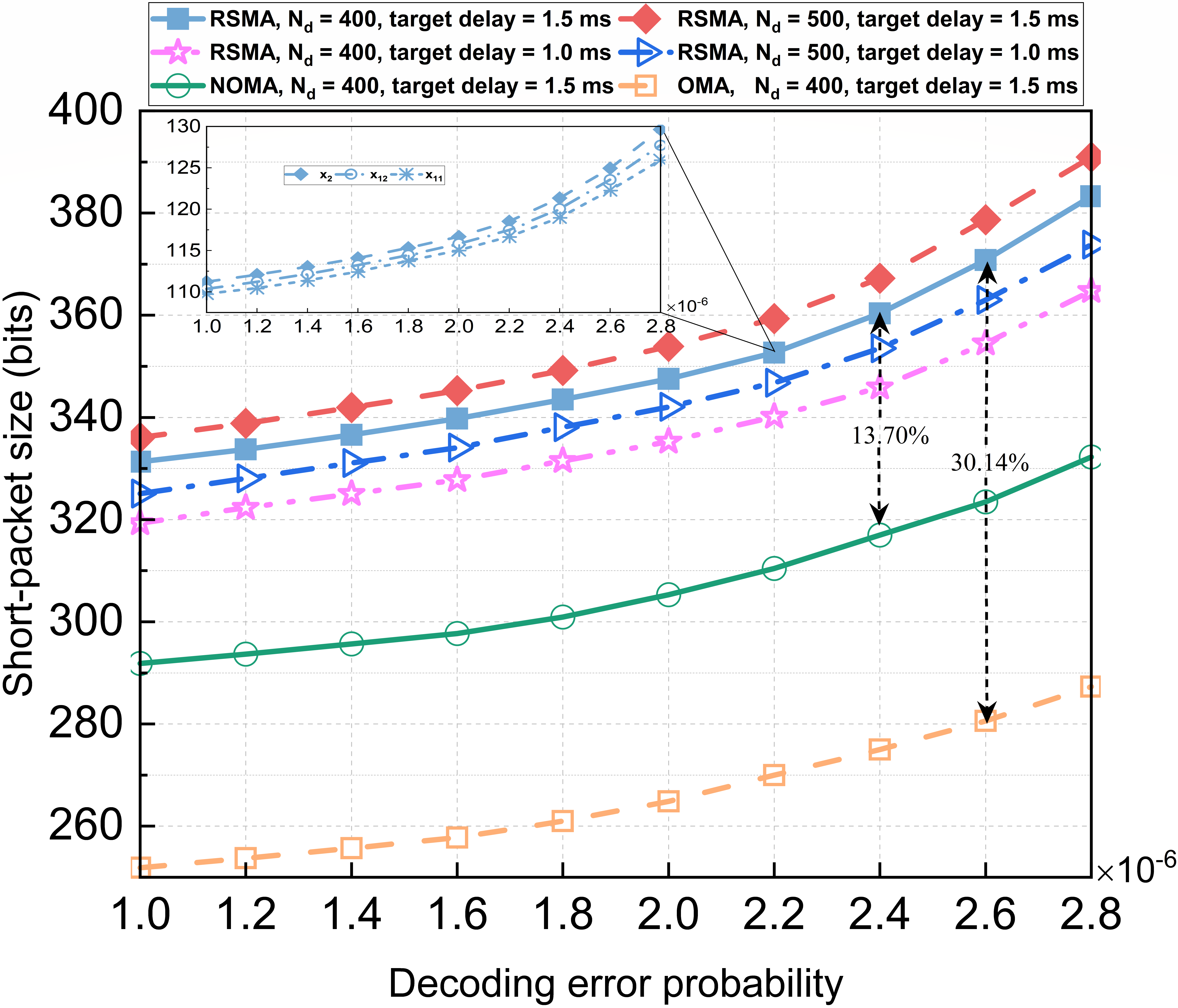}
  }
  \vspace{-0.3em}
  \caption{Performance comparison of different schemes. (a) The maximum short-packet size versus SDVP threshold; (b) The maximum short-packet size versus DEP threshold.}
  \vspace{-0.5em}
\end{figure*}

\vspace{-1em}

\begin{figure*}[t]\label{fig_8}
\centering
  \subfigure[]{
   \includegraphics[scale=0.028]{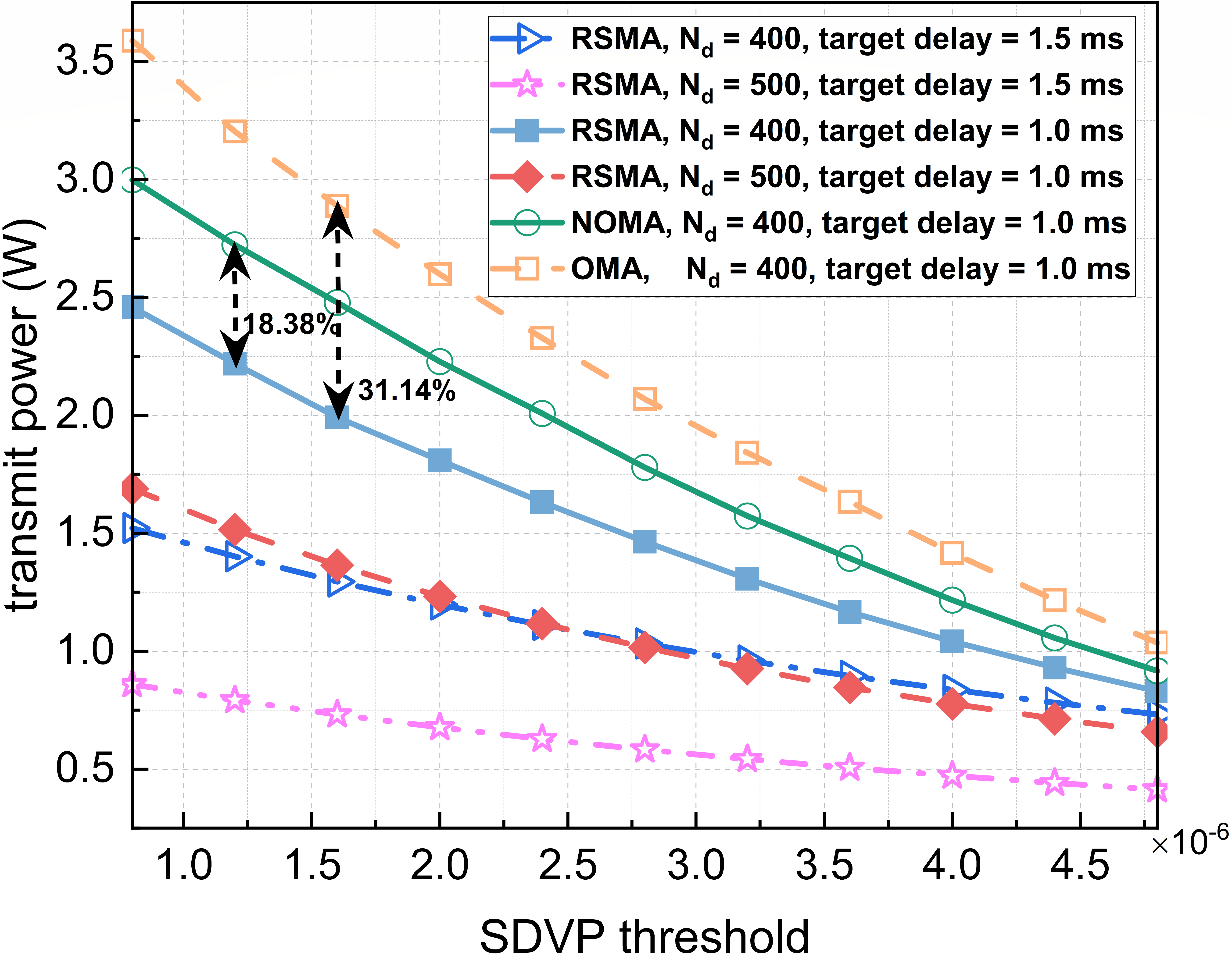}
  }
  \subfigure[]{
   \includegraphics[scale=0.028]{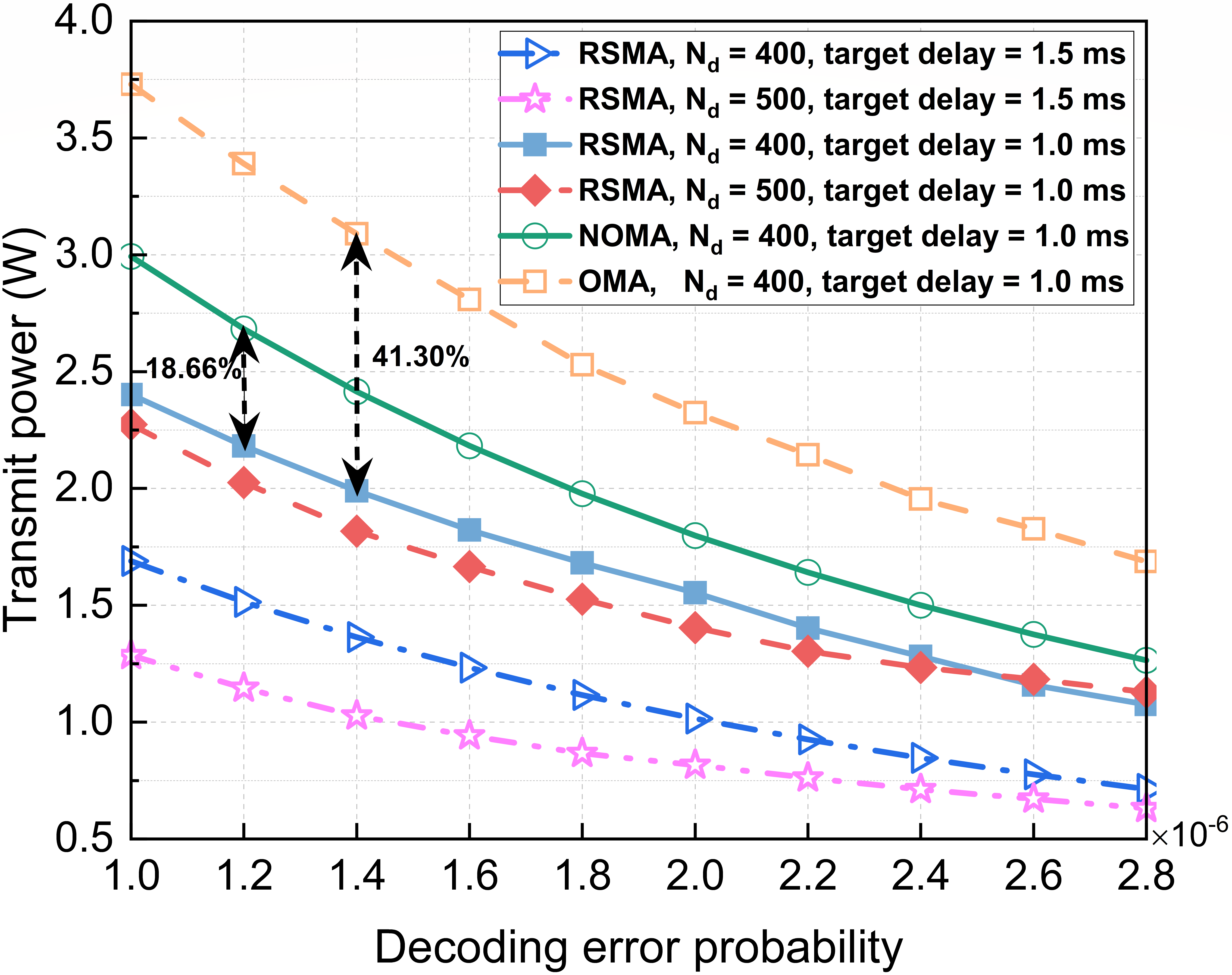}
  }
  \vspace{-0.5em}
  \caption{Performance comparison of different schemes. (a) The minimum transmit power versus SDVP threshold; (b) The minimum transmit power versus DEP threshold.}
  \vspace{-1em}
\end{figure*}

\par In Fig. 7, we investigate the relationship between the maximum short-packet size and the reliability of xURLLC compared to NOMA and OMA schemes across various scenarios. It can be observed that as the SDVP threshold and DEP threshold become more relaxed, the short-packet size supported by our proposed RSMA-xURLLC-IIoT network architecture can be significantly enhanced, which is extremely favorable for accomplishing diverse and complex mission-critical tasks in IIoT networks. Remarkably, $x_{2}$ supports the largest short-packet size, followed by $x_{12}$, and finally $x_{11}$. These distinctions primarily arise from the flexible interference management strategy of the developed RSMA-xURLLC-IIoT network architecture, which splits $x_{1}$ into $x_{11}$ and $x_{12}$ through rate splitting. As a result, $x_{11}$, decoded first, experiences the most severe interference, including co-channel interference from $x_{2}$ and $x_{12}$. Subsequent decoding of $x_{12}$ only experiences co-channel interference from $x_{2}$. In contrast, the last decoded $x_{2}$ remains free from co-channel interference. It ultimately results in distinct co-channel interference characteristics for $x_{2}$, $x_{12}$, and $x_{11}$. Despite more complex interference management in our proposed RSMA-xURLLC-IIoT network architecture, numerical results consistently underscore the remarkable performance advantages. In Fig. 7 (a) and Fig. 7 (b), we provide details of the sequential optimization process performed by the proposed TSSO. Numerical results distinctly demonstrate that, under different SDVP thresholds, our developed RSMA-xURLLC-IIoT network architecture can achieve $14.05\%$ and $31.76\%$ performance gains in terms of augmenting short-packet size compared to NOMA and OMA schemes, respectively. Moreover, as depicted in Fig. 7 (b), we further emphasize that under varying DEP thresholds, our developed RSMA-xURLLC-IIoT network architecture realizes performance enhancements of $13.7\%$ and $30.14\%$ when compared with NOMA and OMA schemes, respectively.

\vspace{-0.7em}

\par As depicted in Fig. 8, we have presented the minimum transmit power required by our developed RSMA-xURLLC-IIoT network architecture compared to NOMA and OMA schemes across various scenarios. It can be observed that as the SDVP threshold and DEP threshold become more lenient, the transmit power required to guarantee the expected SQP performance decreases significantly. Fig. 8 (a) shows that for varying SDVP thresholds, our developed RSMA-xURLLC-IIoT network architecture is capable of conserving $18.38\%$ and $31.14\%$ of the transmit power in terms of improving the power consumption for short-packet transmissions when compared to NOMA and OMA schemes, respectively. As depicted in Fig. 8 (b), with different DEP thresholds, our developed RSMA-xURLLC-IIoT network architecture outperforms NOMA and OMA schemes by achieving $18.66\%$ and $41.30\%$ performance improvements in terms of power consumption for short-packet transmission, respectively.

\par In Fig. 7 and Fig. 8, the fact-filled comparative analyses against NOMA and OMA schemes highlight that our developed RSMA-xURLLC-IIoT network architecture not only facilitates larger short-packet sizes but also achieves superior energy efficiency. This primarily stems from the fact that our developed network architecture empowers multiple access with the potential for concurrent non-orthogonal interference management by partially treating interference as noise while partially decoding it, thereby maximizing interference management optimally. In contrast, NOMA and OMA schemes exhibit disadvantages from two critical perspectives. On the one hand, the NOMA scheme essentially epitomizes an extreme interference management paradigm, which treats interference fully as noise and completely decodes interference \cite{Yuang2023When,clerckx2023primer,mao2022rate}. This severely limits the flexible concurrent non-orthogonal interference management empowerment for multi-access IIoT networks, which in turn leads to poor spectral efficiency and more radio resource overhead. On the other hand, although the OMA scheme can facilitate interference-free multiple access through orthogonal resource allocation, it results in diminished spectrum efficiency. From Fig. 7 and Fig. 8, it can be observed that this deficiency becomes particularly pronounced for xURLLC services with stringent SQP requirements. Therefore, the orthogonal resource allocation in the OMA scheme is unsuitable for multi-access IIoT networks, which intuitively indicates that the OMA scheme can only support relatively small short packets and requires more transmit power to fulfill the expected SQP performance.


\section{Conclusion And Future Outlook}
\par In this paper, we have embarked on a pioneering exploration of xURLLC in Industrial 5.0 and beyond networks by investigating an innovative RSMA-xURLLC-IIoT network architecture. Leveraging SNC theory, we have proposed the SNC-SQP theoretical framework, which is precision-engineered to unveil dependable insights into the SQP analysis for our developed RSMA-xURLLC-IIoT network architecture. Building upon this theoretical framework, we have formulated two SQP-driven optimization problems that hold immense significance in massive-access IIoT networks. To engineer efficient resolution, we have harnessed sequential optimization techniques to propose a low-complexity TTSO algorithm. Through Monte-Carlo methods, we have thoroughly verified the dependability of the proposed SNC-SQP theoretical framework. Additionally, through extensive comparison analyses with prevalent NOMA and OMA schemes, we have further corroborated the superior performance gains achievable by our developed RSMA-xURLLC-IIoT network architecture.

\par For future endeavors, we have deeply contemplated four intimately intertwined perspectives. Firstly, we intend to expand our proposed system model to encompass massive MIMO scenarios, aiming to further elevate network performance \cite{10414053}. Secondly, recognizing age-of-information (AoI) as another critical metric in mission-critical IIoT networks, we plan to incorporate the SQP analysis of AoI for xURLLC traffic into our forthcoming research. Moreover, under FBL regimes, the transmission of machine instructions entails potential security risks \cite{Yuang2023When}. Therefore, the tailored physical layer security mechanisms crafted for our developed RSMA-xURLLC-IIoT network architecture are paramount to guarantee the privacy, stability, and integrity of data \cite{10460318}. Henceforth, another significant facet of our future work entails drafting relevant security metrics for xURLLC, such as false alarm probability and missed detection probability, to combat malicious attacks and data breaches \cite{8643949}. In essence, our proposed SNC-SQP theoretical framework requires further refinement, necessitating the inclusion of additional QoS indicators into the formulated SQP-driven optimization problems, along with corresponding algorithm enhancements.

\vspace{-0.5em}

\begin{appendices}
\section{Proof of Theorem 1}
\setcounter{equation}{0}
\renewcommand\theequation{A-\arabic{equation}}
\vspace{-0.3em}
   According to the definitions of MGF and deconvolution operator $\oslash$ \cite{10382447,fidler2014guide}, we have

\vspace{-0.5em}

   \begin{equation}\label{A1}
     \mathbb{P}\left(W_{q}\left(t\right) > w_{q}^{th} \right) \leq \mathbb{P}\left(\left(A_{q}\oslash S_{q}\right)\left(t+w_{q}^{th},t\right)\geq 0\right).
   \end{equation}

   According to Chernoff's bound for (\ref{A1}) \cite{fidler2014guide,fidler2010survey}, for any random process $X$, the inequality $\mathbb{P}\left(X \geq x\right) \leq e^{-\theta x} \mathbb{M}_{X}\left(\theta\right)$ holds for $\forall x > 0$. By substituting it into \ref{A1}, we can immediately derive that

\vspace{-0.5em}

   \begin{equation}\label{A2}
     \begin{aligned}
       \mathbb{P}\left(W_{q}\left(t\right) > w_{q}^{th} \right) & \leq \mathbb{P}\left(\left(A_{q} \oslash S_{q}\right)\left(t+w_{q}^{th},t\right)\geq 0\right)\\
       & \leq \inf_{\theta > 0} \mathbb{M}_{A_{q} \oslash S_{q}}\left(\theta, t + w_{q}^{th}, t\right).
     \end{aligned}
   \end{equation}

   Based on the definition of MGF, the right-hand side of (\ref{A2}) satisfies

\vspace{-0.5em}

   \begin{equation}\label{A3}
      \begin{aligned}
        & \mathbb{M}_{A_{q} \oslash S_{q}}\left(\theta, t + w_{q}^{th}\right)\\
        & \leq \sum\limits_{u = 0}^{t + w_{q}^{th}} \!\! \mathbb{M}_{A_{q}}\!\left(\theta,u,t\right) \cdot \overline{\mathbb{M}}_{S_{q}}\!\left(\theta,u,t+w_{q}^{th}\right).
      \end{aligned}
   \end{equation}

   Referring to (\ref{e6}) and (\ref{A3}), we can obtain that

\vspace{-0.5em}

   \begin{equation}\label{A4}
      \inf_{\theta > 0} \mathbb{M}_{A_{q} \oslash S_{q}}\left(\theta, t + w_{q}^{th}, t\right) \leq \inf_{\theta > 0} \mathbb{M}_{A_{q} \widehat{\oslash} S_{q}}\left(\theta, t + w_{q}^{th}, t\right).
   \end{equation}

   According to (\ref{A1})-(\ref{A4}), the proof of \textbf{Theorem 1} can be concluded.

\vspace{-0.5em}

\section{Proof of Lemma 1}
\setcounter{equation}{0}
\renewcommand\theequation{B-\arabic{equation}}
\par According to (\ref{e6}), stream $x_{2}$ is immune to co-channel interference, thus $\widehat{\Gamma}_{2} \! \sim \! \mathcal{N}\left(\widehat{\gamma}_{2},\widehat{\sigma}_{q}^{2}\right)$, and the closed-form expression of PDF for $\widehat{\Gamma}_{2}$ can be expressed as follows:
\begin{equation}\label{B1}
   f_{\widehat{\gamma}_{2}}(x) = \frac{1}{\widehat{\sigma}_{2}\sqrt{2\pi}}\exp\left\{-\frac{(x-\widehat{\gamma}_{2})^{2}}{2\widehat{\sigma}_{2}^{2}}\right\}.
\end{equation}

\par Based on (\ref{e7}), stream $x_{1,2}$ experiences co-channel interference from stream $x_{2}$. The c.d.f. of $\widehat{\Gamma}_{1,2}$ can be given as follows:

\vspace{-0.5em}

\begin{equation}\label{B2}
  \begin{aligned}
     F_{\widehat{\Gamma}_{1,2}}\!\!\left(x\right) & = \mathbb{P}\!\left\{\!\frac{\Gamma_{1,2}}{\Gamma_{2} + 1} \!\leq\! x\!\right\} = \mathbb{P}\big\{\Gamma_{1,2} \!\leq \! x\left(1 \!+\! \Gamma_{2}\right)\!\big\}\\
     & = \int_{0}^{\infty} \!\! F_{\Gamma_{1,2}}\!\left(x \!+\! x\Gamma_{2}\right) \cdot f_{\Gamma_{2}}\!\left(\Gamma_{2}\right) \! \d \Gamma_{2}.
  \end{aligned}
\end{equation}

\par Taking the derivative of (\ref{B2}) with respect to $x$, yields:

\vspace{-0.5em}

\begin{equation}\label{B3}
   \begin{aligned}
      & f_{\widehat{\Gamma}_{1,2}}\!\left(x\right) = \frac{\partial F_{\widehat{\Gamma}_{1,2}}\left(x\right)}{\partial x} \\
      & = \int_{0}^{\infty}\left(1 + \Gamma_{2}\right)\cdot f_{\Gamma_{1,2}}\left(x + x\Gamma_{2}\right)\cdot f_{\Gamma_{2}} \left(\Gamma_{2}\right) \d \Gamma_{2}\\
      & = \int_{0}^{\infty} \!\!\! \frac{1+y}{2\pi \widehat{\sigma}_{1,2}\widehat{\sigma}_{2}}\cdot \exp\!\!\left\{\!-\!\left(\!\!\frac{\left(\!\! x+xy-\widehat{\gamma}_{1,2}\right)^{2}}{2\widehat{\sigma}_{1,2}^{2}} \!+\! \frac{\left(y-\widehat{\gamma}_{2}\right)^{2}}{2\widehat{\sigma}_{2}^{2}}\!\!\right)\!\!\right\}\d y.
   \end{aligned}
\end{equation}

\par Since $\Gamma_{1,2}\!\sim\!\mathcal{N}\!\left(\widehat{\gamma}_{1,2},\widehat{\sigma}_{1,2}^{2}\right)$ and $\Gamma_{2}\!\sim\!\mathcal{N}\!\left(\widehat{\gamma}_{2},\widehat{\sigma}_{2}^{2}\right)$, we can obtain that $\Gamma_{1,2} + \Gamma_{2} \!\sim\! \mathcal{N}\!\left(\widehat{\gamma}_{1,2}+\widehat{\gamma}_{2},\widehat{\sigma}_{1,2}^{2}+\widehat{\sigma}_{2}^{2}\right)$. Following a process similar to (\ref{B1})-(\ref{B3}), we can obtain that

\vspace{-0.5em}

\begin{equation}\label{B4}
   \begin{aligned}
     f_{\widehat{\Gamma}_{1,1}}(x) & = \int_{0}^{\infty}\frac{1+y}{2\pi\widehat{\sigma}_{1,1}\widehat{\sigma}_{\sum}}\cdot\\
     & \exp\left\{\!-\!\left(\frac{\left(x+xy-\widehat{\gamma}_{1,1}\right)^{2}}{2\widehat{\sigma}_{1,1}^{2}}\!+\!\frac{\left(y-\widehat{\gamma}_{\sum}^{2}\right)^{2}}{2\widehat{\sigma}_{\sum}^{2}}\!\right)\!\!\right\}\!\!\d y,
   \end{aligned}
\end{equation}
where $\widehat{\sigma}_{\sum}^{2} = \widehat{\sigma}_{1,2}^{2} + \widehat{\sigma}_{2}^{2}$ and $\widehat{\gamma}_{\sum} = \widehat{\gamma}_{1,2} + \widehat{\gamma}_{2}$. So the proof of \textbf{Lemma 1} is concluded.

\vspace{-1em}

\section{Proof of Corollary 1}
\setcounter{equation}{0}
\renewcommand\theequation{C-\arabic{equation}}

\vspace{-0.5em}

\par We first prove the convexity property of the stability condition (\ref{e16}b). The stability condition (\ref{e16}b) can be rewritten as follows

\vspace{-0.5em}

\begin{equation}\label{C1}
   \begin{aligned}
      & S(\theta_{q}) = \mathbb{M}_{a_{q}}(\theta_{q})\cdot \overline{\mathbb{M}}_{s_{q}}(\theta_{q}) \\
      & = \mathbb{E}\left[e^{a_{q}\theta_{q}}\right]\cdot\mathbb{E}\left[e^{-\theta_{q}s_{q}}\right]\\
      & = \mathbb{E}\left[e^{-\theta_{q}\mathcal{X}_{q}}\!\right] < 1,
   \end{aligned}
\end{equation}
where the independent property between $a_{q}$ and $s_{q}$ is utilized, and the variable substitution $\mathcal{X}_{q} = s_{q} - a_{q}$ is exploited. According to the definition of convex functions, we can obtain that $\forall \theta_{q}^{1} \neq \theta_{q}^{2}$ and $\forall 0 \leq \varrho \leq 1$, where $\theta_{q}^{1}, \theta_{q}^{2} \in (0,\theta_{q}^{max})$, the following inequality must hold:

\begin{equation}\label{C2}
   \begin{aligned}
     & \mathbb{E}\left[e^{-\left(\varrho \theta_{q}^{1} + (1-\varrho) \theta_{q}^{2}\right)\mathcal{X}_{q}}\!\right] \\
     & \leq \varrho \mathbb{E}\left[e^{-\varrho\theta_{q}^{1}\mathcal{X}_{q}}\right] + (1-\varrho)\mathbb{E}\left[e^{-\varrho \theta_{q}^{2}\mathcal{X}_{q}}\right].
   \end{aligned}
\end{equation}

\par According to Hölder's inequality, we can derive that
\begin{equation}\label{C3}
   \begin{aligned}
     & \mathbb{E}\left[e^{-\left(\varrho\theta_{q}^{1} + (1-\varrho)\theta_{q}^{2}\right)\mathcal{X}_{q}}\right] = \mathbb{E}\left[\left|e^{-\varrho\theta_{q}^{1}\mathcal{X}_{q}}\cdot e^{-(1-\varrho)\theta_{q}^{2}\mathcal{X}_{q}}\!\right|\right] \\
     & \leq \left(\mathbb{E}\left[\left|e^{-(1-\varrho)\theta_{q}^{2}\mathcal{X}_{q}}\right|^{1/1-\varrho}\right]\right)^{1-\varrho} \cdot \left(\mathbb{E}\left[\left|e^{-\varrho \theta_{q}^{1}\mathcal{X}_{q}}\!\right|^{1/\varrho}\right]\right)^{\varrho} \\
     & = \left(\mathbb{E}\left[e^{-\theta_{q}^{1}\mathcal{X}_{q}}\right]\right)^{\varrho} \cdot \left(\mathbb{E}\left[e^{-\theta_{q}^{2}\mathcal{X}_{q}}\right]\right)^{1-\varrho}\\
     & \overset{(a)}{\leq} \varrho \mathbb{E}\left[e^{-\varrho\theta_{q}^{1}\mathcal{X}_{q}}\right] + (1-\varrho)\mathbb{E}\left[e^{-\varrho \theta_{q}^{2}\mathcal{X}_{q}}\right],
   \end{aligned}
\end{equation}
where the inequality (a) in (\ref{C3}) holds due to the fact that

\begin{equation}\label{C4}
   x_{1}^{\varrho}x_{2}^{1-\varrho} \leq \varrho x_{1} + (1-\varrho) x_{2}, \forall 0 \leq \varrho \leq 1.
\end{equation}

\par As a result, we define the function $f(\varrho) = \varrho x_{1}+(1-\varrho) x_{2} - x_{1}^{\varrho}x_{2}^{1-\varrho}$. $\forall \varrho \in [0,1]$ and $x_{2},x_{2} < 1$, the second-order derivation $f^{''}(\varrho) = - x_{1}^{\varrho}x_{2}^{1-\varrho}\left(\log(x_{2}) - \log(x_{1})\right)^{2} \leq 0$. Since $f(0) = f(1) = 0$ and $f^{''}(\varrho) \leq 0$, it follows that $f(\varrho)$ reaches a local maximum for $\varrho \in [0,1]$. Thus, $\forall \varrho \in [0,1]$, $f(\varrho) \geq 0$ and (\ref{C4}) holds. Therefore, we show that the stability condition is convex.

\par Secondly, we prove the convexity property of the UB-SDVP (\ref{e16}a). According to the convexity property of stability condition, in the feasible domain $(0,\theta_{q}^{max})$, it follows that $1 - S(\theta_{q})$ is a concave and positive function. Hence, its reciprocal is convex, i.e., $\frac{1}{1-S(\theta_{q})}$ \cite{boyd2004convex}. According to the definition of convex function, $\forall \theta_{q}^{1}, \theta_{q}^{2} \in (0,\theta_{q}^{max})$ and $0 \leq \varrho \leq 1$, we have

\vspace{-0.3em}

\begin{equation}\label{C5}
   \begin{aligned}
      & \frac{1}{1 - \mathbb{E}\left[e^{\left(\varrho \theta_{q}^{1} + (1-\varrho)\theta_{q}^{2}\right)a_{q}}\right]\cdot \mathbb{E}\left[e^{-\left(\varrho\theta_{q}^{1}+ (1-\varrho)\theta_{q}^{2}\right)s_{q}}\right]} \leq \\
      & \frac{\varrho}{1-\mathbb{E}\left[e^{\theta_{q}^{1}a_{q}}\right] \cdot  \mathbb{E}\left[e^{-\theta_{q}^{1}s_{q}}\right]} + \frac{1-\varrho}{1-\mathbb{E}\left[e^{\theta_{q}^{2}a_{q}}\right]\cdot \mathbb{E}\left[e^{-\theta_{q}^{2}s_{q}}\right]}.
   \end{aligned}
\end{equation}

\par By multiplying both the left and right sides of the inequality (\ref{C5}) by $\left(\!\mathbb{E}\!\!\left[e^{-(\varrho \theta_{q}^{1} + (1-\varrho)\theta_{q}^{2}\!)s_{q}}\!\right]\!\right)^{w_{q}^{\ast}}$, we derive that

\begin{equation}\label{C6}
  \begin{aligned}
    & \frac{\left(\mathbb{E}\left[e^{-(\!\varrho \theta_{q}^{1} + (1-\varrho)\theta_{q}^{2})s_{q}}\right]\right)^{w_{q}^{\ast}}}{1- \mathbb{E}\left[e^{\left(\varrho\theta_{q}^{1} + (1-\varrho)\theta_{q}^{2}\right)a_{q}}\right]\cdot \mathbb{E}\left[e^{-\left(\varrho\theta_{q}^{1} + (1-\varrho)\theta_{q}^{2}\right)s_{q}}\right]} \leq \\
    & \frac{\varrho \left(\mathbb{E}\left[e^{-(\varrho \theta_{q}^{1} + (1-\varrho)\theta_{q}^{2})s_{q}}\right]\right)^{w_{q}^{\ast}}}{1 -\mathbb{E}\left[e^{\theta_{q}^{1}a_{q}}\right]\cdot \mathbb{E}\left[e^{-\theta_{q}^{1}s_{q}}\right]} + \frac{1-\varrho}{1-\mathbb{E}\left[e^{\theta_{q}^{2}a_{q}}\right]\cdot \mathbb{E}\left[e^{-\theta_{q}^{2}s_{q}}\!\right]}\\
    & \times \left(\mathbb{E}\left[e^{-(\varrho \theta_{q}^{1} + (1-\varrho)\theta_{q}^{2})s_{q}}\right]\right)^{w_{q}^{\ast}}
  \end{aligned}
\end{equation}

\par By using Hölder’s inequality, we can derive that

\begin{equation}\label{C7}
   \begin{aligned}
     & \left(\!\mathbb{E}\!\left[e^{-(\varrho \theta_{q}^{1} + (1-\varrho)\theta_{q}^{2})s_{q}}\!\!\right]\!\right)^{\!\!w_{q}^{\ast}}\!=\! \left(\!\!\mathbb{E}\!\!\left[\left|e^{-\varrho \theta_{q}^{1}s_{q}}\!\right|\right] \! \cdot \! \mathbb{E}\!\!\left[\left|e^{-(1-\varrho) \theta_{q}^{2}s_{q}}\right|\right]\right)^{\!\!w_{q}^{\ast}}\\
     & \leq \left(\mathbb{E}\left[\left|e^{-(1-\varrho) \theta_{q}^{2}s_{q}}\right|^{1/1-\varrho}\right]\right)^{w_{q}^{\ast}(1-\varrho)} \!\! \times \!\! \left(\mathbb{E}\left[\left|e^{-\varrho \theta_{q}^{1}s_{q}}\right|^{1/\varrho}\right]\right)^{w_{q}^{\ast}\varrho}\\
     & = \left(\mathbb{E}\left[\left|e^{-\theta_{q}^{1}s_{q}}\right|\right]\right)^{w_{q}^{\ast}\varrho} \cdot \left(\mathbb{E}\left[\left|e^{-\theta_{q}^{2}s_{q}}\right|\right]\right)^{w_{q}^{\ast}(1-\varrho)}\\
     & \leq \left(\mathbb{E}\left[\left|e^{-\theta_{q}^{1}s_{q}}\right|\right] \cdot \mathbb{E}\left[\left|e^{-\theta_{q}^{2}s_{q}}\right|\right]\right)^{w_{q}^{\ast}},
   \end{aligned}
\end{equation}

\par Due to the fact that $\left(\mathbb{E}\left[\left|e^{-\theta_{q}s_{q}}\right|\right]\right)^{w_{q}^{\ast}} > 0$. Then, the following inequality applies to the right side of (\ref{C6}), and we can derive that
\begin{equation}\label{C8}
\setstretch{0.95}
  \begin{aligned}
     & \frac{\varrho \! \left(\mathbb{E}\!\!\left[e^{-(\varrho \theta_{q}^{1} + (1-\varrho)\theta_{q}^{2})s_{q}}\!\right]\right)^{w_{q}^{\ast}}}{1 \!-\! \mathbb{E}\!\! \left[e^{\theta_{q}^{1}a_{q}}\right] \! \cdot \! \mathbb{E}\!\!\left[e^{-\theta_{q}^{1}s_{q}}\!\right]} \!+\! \frac{(1\!-\!\varrho)\!\left(\mathbb{E}\!\!\left[e^{-(\varrho \theta_{q}^{1} + (1-\varrho)\theta_{q}^{2})s_{q}}\!\right]\!\right)^{w_{q}^{\ast}}}{1\!-\!\mathbb{E}\!\left[e^{\theta_{q}^{2}a_{q}}\right]\!\cdot\! \mathbb{E}\!\left[e^{-\theta_{q}^{2}s_{q}}\!\right]} \\
     & \leq \!\! \Biggl(\!\! \frac{\varrho \! \left(\mathbb{E}\left[\big|e^{-\theta_{q}^{1} s_{q}}\big|\right] \! \cdot \! \mathbb{E}\left[\big|e^{-\theta_{q}^{2}s_{q}}\big|\right]\right)^{w_{q}^{\ast}}}{1 -\mathbb{E}\left[e^{\theta_{q}^{1} a_{q}}\right]\cdot \mathbb{E}\left[ e^{-\theta_{q}^{1} s_{q}}\right]} \! + \! \frac{(1\!-\!\varrho)}{1\!-\!\mathbb{E}\!\!\left[e^{\theta_{q}^{2}a_{q}}\!\right]\!\cdot\! \mathbb{E}\!\!\left[\!e^{-\theta_{q}^{2}s_{q}}\!\right]}\\
     & \times \!\! \left(\mathbb{E}\!\!\left[\big|e^{-\theta_{q}^{1}s_{q}}\big|\right] \! \cdot \! \mathbb{E}\!\!\left[\big|e^{-\theta_{q}^{2}s_{q}}\big|\right]\right)^{\!\!w_{q}^{\ast}} \!\!\! \Biggl) \leq \frac{\varrho \left(\mathbb{E}\left[\big|e^{-\theta_{q}^{1}s_{q}}\big|\right]\right)^{w_{q}^{\ast}}}{1-\mathbb{E}\left[e^{\theta_{q}^{1}a_{q}}\right] \cdot \mathbb{E}\left[e^{-\theta_{q}^{1}s_{q}}\right]} \\
     & + \frac{(1-\varrho)\left(\mathbb{E}\left[\big|e^ {-\theta_{q}^{1}s_{q}}\big|\right]\cdot \mathbb{E}\left[\big|e^{-\theta_{q}^{2}s_{q}}\big|\right]\right)^{w_{q}^{\ast}}.}{1-\mathbb{E}\left[e^{\theta_{q}^{2}a_{q}}\right]\cdot \mathbb{E}\left[e^{-\theta_{q}^{2}s_{q}}\right]}
  \end{aligned}
\end{equation}

\par Combining (\ref{C6}), (\ref{C7}), and (\ref{C8}), as well as the definition of convex functions, we can finally prove that the kernel function $K_{q}(\theta_{q},-w_{q}^{\ast})$ is convex. So \textbf{Corollary 1} can be concluded.
\end{appendices}

\footnotesize
\bibliographystyle{IEEEtran}
\bibliography{IEEEabrv,reference_globecom}

\begin{thebibliography}{10}
\providecommand{\url}[1]{#1}
\csname url@samestyle\endcsname
\providecommand{\newblock}{\relax}
\providecommand{\bibinfo}[2]{#2}
\providecommand{\BIBentrySTDinterwordspacing}{\spaceskip=0pt\relax}
\providecommand{\BIBentryALTinterwordstretchfactor}{4}
\providecommand{\BIBentryALTinterwordspacing}{\spaceskip=\fontdimen2\font plus
\BIBentryALTinterwordstretchfactor\fontdimen3\font minus
  \fontdimen4\font\relax}
\providecommand{\BIBforeignlanguage}[2]{{%
\expandafter\ifx\csname l@#1\endcsname\relax
\typeout{** WARNING: IEEEtran.bst: No hyphenation pattern has been}%
\typeout{** loaded for the language `#1'. Using the pattern for}%
\typeout{** the default language instead.}%
\else
\language=\csname l@#1\endcsname
\fi
#2}}
\providecommand{\BIBdecl}{\relax}
\BIBdecl

\bibitem{park2022extreme}
J.~Park \emph{et~al.}, ``Extreme ultra-reliable and low-latency
  communication,'' \emph{Nat. Electron.}, vol.~5, no.~3, pp. 133--141, Mar.
  2022.

\bibitem{she2021tutorial}
C.~She \emph{et~al.}, ``A tutorial on ultrareliable and low-latency
  communications in {6G}: Integrating domain knowledge into deep learning,''
  \emph{Proc. IEEE}, vol. 109, no.~3, pp. 204--246, Mar. 2021.

\bibitem{10382447}
Y.~Chen \emph{et~al.}, ``Statistical {QoS} provisioning analysis and
  performance optimization in x{URLLC}-enabled massive {MU-MIMO} networks: A
  stochastic network calculus perspective,'' \emph{IEEE Trans. Wireless
  Commun.}, pp. 1--1, 2024.

\bibitem{Yuang2023When}
Y.~Chen \emph{et~al.}, ``\BIBforeignlanguage{English}{When {xURLLC} meets
  {NOMA}: A stochastic network calculus perspective},''
  \emph{\BIBforeignlanguage{English}{IEEE Commun. Mag.}}, Jul. 2023.

\bibitem{chen2024enhancing}
Y.~Chen \emph{et~al.}, ``Enhancing {xURLLC} with {RSMA}-assisted massive-{MIMO}
  networks: Performance analysis and optimization,'' \emph{arXiv preprint
  arXiv:2402.16027}, 2024.

\bibitem{10021621}
Y.~Liu \emph{et~al.}, ``Deep reinforcement learning-based grant-free {NOMA}
  optimization for {mURLLC},'' \emph{IEEE Trans. Commun.}, vol.~71, no.~3, pp.
  1475--1490, 2023.

\bibitem{xian2023advanced}
W.~Xian \emph{et~al.}, ``Advanced manufacturing in industry 5.0: A survey of
  key enabling technologies and future trends,'' \emph{IEEE Trans. Ind.
  Informat.}, 2023.

\bibitem{khoshnevisan20195g}
M.~Khoshnevisan \emph{et~al.}, ``5{G} industrial networks with {CoMP} for
  {URLLC} and time sensitive network architecture,'' \emph{IEEE J. Sel. Areas
  Commun.}, vol.~37, no.~4, pp. 947--959, 2019.

\bibitem{polyanskiy2010channel}
Y.~Polyanskiy \emph{et~al.}, ``Channel coding rate in the finite blocklength
  regime,'' \emph{IEEE Trans. Inf. Theory}, vol.~56, no.~5, pp. 2307--2359, May
  2010.

\bibitem{yang2014quasi}
W.~Yang \emph{et~al.}, ``Quasi-static multiple-antenna fading channels at
  finite blocklength,'' \emph{IEEE Trans. Inf. Theory}, vol.~60, no.~7, pp.
  4232--4265, Jul. 2014.

\bibitem{mao2022rate}
Y.~Mao \emph{et~al.}, ``Rate-splitting multiple access: Fundamentals, survey,
  and future research trends,'' \emph{IEEE Commun. Surv. Tutorials}, 2022.

\bibitem{clerckx2023primer}
B.~Clerckx \emph{et~al.}, ``A primer on rate-splitting multiple access:
  Tutorial, myths, and frequently asked questions,'' \emph{IEEE J. Sel. Areas
  Commun.}, 2023.

\bibitem{mishra2022rate}
A.~Mishra \emph{et~al.}, ``Rate-splitting multiple access for {6G}—part i:
  Principles, applications and future works,'' \emph{IEEE Commun. Lett.},
  vol.~26, no.~10, pp. 2232--2236, 2022.

\bibitem{rimoldi1996rate}
B.~Rimoldi and R.~Urbanke, ``A rate-splitting approach to the gaussian
  multiple-access channel,'' \emph{IEEE Trans. Inf. Theory}, vol.~42, no.~2,
  pp. 364--375, 1996.

\bibitem{li2024synergizing}
H.~Li \emph{et~al.}, ``Synergizing beyond diagonal reconfigurable intelligent
  surface and rate-splitting multiple access,'' \emph{IEEE Trans. Wireless
  Commun.}, 2024.

\bibitem{dizdar2024rate}
O.~Dizdar and S.~Wang, ``Rate-splitting multiple access for semantic-aware
  networks: an age of incorrect information perspective,'' \emph{IEEE Wireless
  Commun. Lett.}, 2024.

\bibitem{lei2024secure}
H.~Lei \emph{et~al.}, ``On secure mm{W}ave {RSMA} systems,'' \emph{IEEE
  Internet Things J.}, 2024.

\bibitem{10430407}
L.~Qin \emph{et~al.}, ``Joint transmission and resource optimization in
  {NOMA}-assisted {IoVT} with mobile edge computing,'' \emph{IEEE Trans. Veh.
  Technol.}, pp. 1--16, 2024.

\bibitem{khisa2024power}
S.~Khisa \emph{et~al.}, ``Power allocation and beamforming design for uplink
  rate-splitting multiple access with user cooperation,'' \emph{IEEE Trans.
  Veh. Technol.}, 2024.

\bibitem{sarker2023uplink}
M.~Sarker and A.~O. Fapojuwo, ``Uplink power allocation for {RSMA}-aided
  user-centric cell-free massive {MIMO} systems,'' in \emph{2023 IEEE 97th
  Vehicular Technology Conference (VTC2023-Spring)}.\hskip 1em plus 0.5em minus
  0.4em\relax IEEE, 2023, pp. 1--5.

\bibitem{9852986}
O.~Abbasi \emph{et~al.}, ``Transmission scheme, detection and power allocation
  for uplink user cooperation with {NOMA} and {RSMA},'' \emph{IEEE Trans.
  Wireless Commun.}, vol.~22, no.~1, pp. 471--485, 2023.

\bibitem{chen2023streaming}
Y.~Chen \emph{et~al.}, ``Streaming 360-degree {VR} video with statistical {QoS}
  provisioning in {mmWave} networks from delay and rate perspectives,''
  \emph{arXiv preprint arXiv:2305.07935}, 2023.

\bibitem{bennis2018ultrareliable}
M.~Bennis \emph{et~al.}, ``Ultrareliable and low-latency wireless
  communication: Tail, risk, and scale,'' \emph{Proc. IEEE}, vol. 106, no.~10,
  pp. 1834--1853, Oct. 2018.

\bibitem{al2014network}
H.~Al-Zubaidy \emph{et~al.}, ``Network-layer performance analysis of multihop
  fading channels,'' \emph{IEEE/ACM Trans. Networking}, vol.~24, no.~1, pp.
  204--217, Feb. 2014.

\bibitem{fidler2010survey}
M.~Fidler, ``Survey of deterministic and stochastic service curve models in the
  network calculus,'' \emph{IEEE Commun. Surv. Tutorials}, vol.~12, no.~1, pp.
  59--86, First Quarter 2010.

\bibitem{fidler2014guide}
M.~Fidler and A.~Rizk, ``A guide to the stochastic network calculus,''
  \emph{IEEE Commun. Surv. Tutorials}, vol.~17, no.~1, pp. 92--105,
  Firstquarter 2014.

\bibitem{singh2023rsma}
S.~K. Singh \emph{et~al.}, ``{RSMA} for hybrid {RIS-UAV}-aided {Full-Duplex}
  communications with finite blocklength codes under imperfect {SIC},''
  \emph{IEEE Trans. Wireless Commun.}, 2023.

\bibitem{ou2022resource}
X.~Ou \emph{et~al.}, ``Resource allocation in {MU-MISO} rate-splitting multiple
  access with {SIC} errors for {URLLC} services,'' \emph{IEEE Trans. Commun.},
  2022.

\bibitem{wang2023flexible}
Y.~Wang \emph{et~al.}, ``Flexible rate-splitting multiple access with finite
  blocklength,'' \emph{IEEE J. Sel. Areas Commun.}, vol.~41, no.~5, pp.
  1398--1412, 2023.

\bibitem{kurma2022urllc}
S.~Kurma \emph{et~al.}, ``{URLLC}-based cooperative industrial {IoT} networks
  with nonlinear energy harvesting,'' \emph{IEEE Trans. Ind. Informat.},
  vol.~19, no.~2, pp. 2078--2088, 2022.

\bibitem{muhammad2021mission}
I.~Muhammad \emph{et~al.}, ``Mission effective capacity—a novel dependability
  metric: A study case of multiconnectivity-enabled {URLLC} for {IIoT},''
  \emph{IEEE Trans. Ind. Informat.}, vol.~18, no.~6, pp. 4180--4188, 2021.

\bibitem{lien2022intelligent}
S.-Y. Lien and D.-J. Deng, ``Intelligent session management for {URLLC} in {5G}
  open radio access network: A deep reinforcement learning approach,''
  \emph{IEEE Trans. Ind. Informat.}, vol.~19, no.~2, pp. 1844--1853, 2022.

\bibitem{10330667}
J.~Hu \emph{et~al.}, ``Low-complexity resource allocation for uplink {RSMA} in
  future {6G} wireless networks,'' \emph{IEEE Wireless Commun. Lett.}, vol.~13,
  no.~2, pp. 565--569, 2024.

\bibitem{9970313}
J.~Xu \emph{et~al.}, ``Rate-splitting multiple access for short-packet uplink
  communications: A finite blocklength analysis,'' \emph{IEEE Commun. Lett.},
  vol.~27, no.~2, pp. 517--521, 2023.

\bibitem{10323296}
O.~L.~A. López \emph{et~al.}, ``Statistical tools and methodologies for
  ultrareliable low-latency communication—a tutorial,'' \emph{Proc. IEEE},
  vol. 111, no.~11, pp. 1502--1543, 2023.

\bibitem{9887634}
O.~Adamuz-Hinojosa \emph{et~al.}, ``A stochastic network calculus ({SNC})-based
  model for planning {B5G} u{RLLC} {RAN} slices,'' \emph{IEEE Trans. Wireless
  Commun.}, vol.~22, no.~2, pp. 1250--1265, 2023.

\bibitem{10220199}
P.~Cui \emph{et~al.}, ``End-to-end delay performance analysis of industrial
  internet of things: A stochastic network calculus perspective,'' \emph{IEEE
  Internet Things J.}, vol.~11, no.~3, pp. 5374--5387, 2024.

\bibitem{10445487}
C.~Wu \emph{et~al.}, ``Cross-layer optimization for statistical {QoS} provision
  in {C-RAN} with finite-length coding,'' \emph{IEEE Trans. Commun.}, pp. 1--1,
  2024.

\bibitem{du2004sequential}
X.~Du \emph{et~al.}, ``Sequential optimization and reliability assessment
  method for efficient probabilistic design,'' \emph{J. Mech. Des.}, vol. 126,
  no.~2, pp. 225--233, 2004.

\bibitem{10414053}
G.~Zheng \emph{et~al.}, ``Joint hybrid precoding and rate allocation for {RSMA}
  in near-field and far-field massive {MIMO} communications,'' \emph{IEEE
  Wireless Commun. Lett.}, pp. 1--1, 2024.

\bibitem{10460318}
L.~Qin \emph{et~al.}, ``Energy-efficient blockchain-enabled user-centric mobile
  edge computing,'' \emph{IEEE Trans. Cogn. Commun. Netw.}, pp. 1--1, 2024.

\bibitem{8643949}
H.~Forssell \emph{et~al.}, ``Physical layer authentication in mission-critical
  {MTC} networks: A security and delay performance analysis,'' \emph{IEEE J.
  Sel. Areas Commun.}, vol.~37, no.~4, pp. 795--808, 2019.

\bibitem{boyd2004convex}
S.~P. Boyd and L.~Vandenberghe, \emph{Convex optimization}.\hskip 1em plus
  0.5em minus 0.4em\relax Cambridge university press, 2004.

\end{thebibliography}
\end{document}